\documentclass{article}

\usepackage[english]{babel}

\usepackage[a4paper,top=2cm,bottom=2cm,left=3cm,right=3cm,marginparwidth=1.75cm]{geometry}

\usepackage[colorlinks=true, allcolors=blue]{hyperref}
\usepackage{times}

\usepackage{soul}
\usepackage{url}
\usepackage[utf8]{inputenc}
\usepackage[small]{caption}
\usepackage{graphicx}
\usepackage{amsmath,amssymb,amsthm}
\usepackage{amsfonts}
\usepackage{booktabs}
\usepackage{color}
\usepackage{xfrac}
\usepackage{cleveref}
\usepackage[ruled,linesnumbered,vlined]{algorithm2e}
\usepackage{tikz}
\usepackage{pgfplots}
\pgfplotsset{compat=1.17}
\usetikzlibrary{patterns}
\usepackage{pgfplotstable}
\pgfplotsset{
	name nodes near coords/.style={
		every node near coord/.append style={
			name=#1-\coordindex,
			alias=#1-last,
		},
	},
	name nodes near coords/.default=coordnode
}

\newtheorem{claim}{Claim}[section]
\newtheorem{theorem}{Theorem}[section]
\newtheorem{corollary}[theorem]{Corollary}
\newtheorem{lemma}[theorem]{Lemma}

\newtheorem{example}[theorem]{Example}
\newtheorem{definition}[theorem]{Definition}

\newtheorem{invariant}[theorem]{Invariant}

\newenvironment{proofof}[1]{{\vspace*{5pt} \noindent\bf Proof of #1:  }}{\hfill\rule{2mm}{2mm}\vspace*{5pt}}

\makeatletter
\newtheorem*{rep@theorem}{\rep@title}
\newcommand{\newreptheorem}[2]{%
	\newenvironment{rep#1}[1]{%
		\def\rep@title{#2 \ref{##1}}%
		\begin{rep@theorem}}%
		{\end{rep@theorem}}}
\makeatother

\newreptheorem{lemma}{Lemma}
\newreptheorem{claim}{Claim}

\newcommand{\bR}{{\mathbb{R}}}

\newcommand{\MPB}{\mathsf{MPB}}

\newcommand{\APS}{\mathsf{APS}}
\newcommand{\bX}{\mathbf{X}}
\newcommand{\bw}{\mathbf{w}}
\newcommand{\bc}{\mathbf{c}}
\newcommand{\bv}{\mathbf{v}}
\newcommand{\bp}{\mathbf{p}}

\newcommand{\POF}{\mathsf{PoF}}
\newcommand{\opt}{{\rm opt}}
\renewcommand{\sc}{{\rm sc}}

\DeclareMathOperator*{\argmax}{argmax}

\title{Weighted EF1 Allocations for Indivisible Chores}

\author{Xiaowei Wu
	\thanks{IOTSC, University of Macau. \{xiaoweiwu,yc27429,yc17423\}@um.edu.mo. The authors are ordered alphabetically.}
	\and Cong Zhang $^*$
	\and Shengwei Zhou $^*$}

\date{}

\begin{document}
	\pagestyle{plain}
	\maketitle
	
	\begin{abstract}
		We study how to fairly allocate a set of indivisible chores to a group of agents, where each agent $i$ has a non-negative weight $w_i$ that represents its obligation for undertaking the chores.
		We consider the fairness notion of {\em weighted envy-freeness up to one item} (WEF1) and propose an efficient picking sequence algorithm for computing WEF1 allocations.
		Our analysis is based on a natural and powerful continuous interpretation for the picking sequence algorithms in the weighted setting, which might be of independent interest.
		Using this interpretation, we establish the necessary and sufficient conditions under which picking sequence algorithms can guarantee other fairness notions in the weighted setting.
		We also study the existence of fair and efficient allocations and propose efficient algorithms for the computation of WEF1 and PO allocations for the bi-valued instances.
		Our result generalizes that of Garg et al.~\cite{conf/aaai/GargMQ22} and Ebadian et al.~\cite{conf/atal/EbadianP022} to the weighted setting.
		Our work also studies the price of fairness for WEF1, and the implications of WEF1 to other fairness notions.
	\end{abstract}
	
	\section{Introduction}\label{sec:intro}
	As a classic problem that can be traced back to 1948~\cite{steihaus1948problem}, fair allocation has received much attention in the past decades, in the fields of computer science, economics, and mathematics.
	While the traditional study of fair allocation focused on divisible items~\cite{foley1967resource,alon1987splitting,edward1999rental}, there is an increasing attention to the fair allocation of indivisible items in recent years.
	In this problem, our goal is to allocate a set $M$ of $m$ indivisible items to a set $N$ of $n$ agents, where agents may have different valuation functions on the items.
	An allocation is defined as an $n$-partition $\bX = (X_1, X_2, \ldots, X_n)$ of the items, where $X_i \cap X_j = \emptyset$ for all $i\neq j$ and $\cup_{i\in N} X_i = M$.
	We say that $X_i$ is the bundle assigned to agent $i\in N$.
	Depending on whether the agents have positive or negative values on the items, there are two lines of research, one for the allocation of goods and the other for chores.
	In this work we focus on the allocation of chores, in which the agents have negative values on the items.
	For convenience of notation, we assume that each agent has a cost function that assigns a non-negative cost to each bundle of items.
	We assume that all cost functions are additive.
	
	Two of the most well studied fairness notions are \emph{envy-freeness} (EF)~\cite{foley1967resource} and \emph{proportionality} (PROP)~\cite{steihaus1948problem}.
	An allocation is PROP if every agent receives a bundle with cost at most her proportional cost of all items, i.e., $c_i(X_i) \leq \frac{1}{n}\cdot c_i(M)$ for all $i\in N$.
	An allocation is EF if under the allocation no agent wants to exchange her bundle of items with some other agent to decrease her cost, i.e., $c_i(X_i) \leq c_i(X_j)$ for all $i\neq j$.
	Observe that every EF allocation is PROP.
	For the allocation of divisible items, EF allocations and PROP allocations have been proven to exist~\cite{alon1987splitting,edward1999rental}.
	However, PROP allocations (and thus EF allocations) are not guaranteed to exist when items are indivisible, e.g., consider the allocation of a single item to two agents having non-zero cost on the item.
	Therefore, researchers have proposed several relaxation of these fairness notions.
	\emph{Envy-freeness up to one item} (EF1) is one of the most well studied relaxations of envy-freeness, and is introduced by Lipton et al.~\cite{conf/sigecom/LiptonMMS04}.
	An allocation is said to be EF1 if the envy between any two agents can be eliminated after removing one item from the bundle of the envious agent.
	It has been shown that EF1 allocations always exist for the cases of goods~\cite{conf/sigecom/LiptonMMS04}, chores, and even mixed items~\cite{conf/approx/BhaskarSV21,journals/aamas/AzizCIW22}.
	Another relaxation is the \emph{envy-freeness up to any item} (EFX) proposed by Caragiannis et al.~\cite{journals/teco/CaragiannisKMPS19}, which is stronger than EF1 and requires that envy can be eliminated after removing any item from the bundle.
	However, unlike EF1, EFX allocations are guaranteed to exist only for some very special cases~\cite{conf/sigecom/ChaudhuryGM20,journals/tcs/AmanatidisMN20,conf/www/0037L022,conf/ijcai/0002022}.
	Whether EFX allocations exist in general remains one of the biggest open problems.
	Similarly, \emph{proportionality up to one item} (PROP1)~\cite{conf/sigecom/ConitzerF017} and \emph{proportionality up to any item} (PROPX)~\cite{journals/orl/AzizMS20} are two well-known relaxations of proportionality.
	PROP1 allocations have been proven to exist, for both goods~\cite{conf/sigecom/ConitzerF017} and chores~\cite{journals/orl/AzizMS20}, while PROPX allocations might not exist in the cases of goods~\cite{journals/orl/AzizMS20} and can be found in polynomial time for chores~\cite{conf/www/0037L022}.
	For a more detailed review of the fair allocation problem, please refer to the recent surveys by Amanatidis et al.~\cite{journals/corr/abs-2208-08782} and Aziz et al.~\cite{journals/corr/abs-2202-08713}.

	\paragraph{Weighted Setting.}
	While traditional fair allocation problem focuses on the case where agents have the equal obligations, in the real world, it often happens that agents are not equally obliged.
	For example, a person in leadership position is naturally expected to undertake more responsibilities for finishing the set of tasks.
	To model these applications, the \emph{weighted} (or \emph{asymmetric}) setting is proposed~\cite{journals/teco/ChakrabortyISZ21,conf/www/0037L022}.
	In the weighted setting, each agent $i\in N$ has a weight $w_i > 0$ that represents the obligation of agent $i$ on the chores, and we have $\sum_{i\in N} w_i = 1$.
	Note that in the unweighted case we have $w_i = 1/n$ for all $i\in N$.
	Chakraborty et al.~\cite{journals/teco/ChakrabortyISZ21} introduce the \emph{weighted envy-freeness up to one item} (WEF1) for the allocation of goods and show that WEF1 allocations always exist and can be computed in polynomial time.
	Li et al.~\cite{conf/www/0037L022} consider the allocation of chores and propose an algorithm that computes WEF1 allocations for the \emph{identical ordering} (IDO) instance and WPROPX allocations for general instances.
	Whether WEF1 allocations exist for general instances remains unknown, and has been proposed as an open problem in several existing works~\cite{journals/teco/ChakrabortyISZ21,conf/www/0037L022,journals/corr/abs-2202-08713,journals/mst/BeiLMS21}.
	In this work, we answer this open problem affirmatively.

	\paragraph{Efficiency.}
	Besides fairness, efficiency is another important measurement for the quality of allocations.
	Unfortunately efficiency and fairness are often competing with each other, e.g., many of the fair allocations give every bad efficiency guarantees.
	Therefore, the existence of fair and efficient allocations has recently drawn a significant attention.
	Popular efficiency measurements include the social cost $\sum_{i\in N} c_i(X_i)$ and the Pareto optimality.
	An allocation is said to be \emph{Pareto optimal} (PO) if there does not exist another allocation that can decrease the cost of some agent without increasing the costs of other agents.
	For the allocation of goods, Caragiannis et al.~\cite{journals/teco/CaragiannisKMPS19} show the allocation maximizing \emph{Nash social welfare} is EF1 and PO for unweighted agents.
	For the allocation of chores, the existence of EF1 and PO allocations is still a major open problem.
	Allocations that are EF1 and PO are known to exist only for some special cases, e.g. two agents~\cite{journals/aamas/AzizCIW22} and \emph{bi-valued} instances~\cite{conf/aaai/GargMQ22,conf/atal/EbadianP022}.
	Besides, the \emph{price of fairness} (PoF) that measures the loss in social welfare/cost due to the fairness constraints has also received an increasing attention~\cite{journals/ior/BertsimasFT11,journals/mst/CaragiannisKKK12,journals/mst/BeiLMS21,conf/wine/BarmanB020,conf/atal/SunCD21}.

	\subsection{Our Contribution}
	
	In this paper, we consider the existence, computation and efficiency of WEF1 allocations for indivisible chores.
	We first show that WEF1 allocations always exist for chores, by proposing a polynomial time algorithm based on the weighted picking sequence protocols.
	
	\smallskip
	\noindent
	{\bf Result 1} (Theorem~\ref{thm:compute-wef1}) {\bf .}
	{\em For the allocation of chores to weighted agents, there exists a polynomial time algorithm that computes WEF1 allocations.}
	\smallskip
	
	Similar to existing works~\cite{journals/teco/ChakrabortyISZ21,conf/www/0037L022,azizbest,conf/aaai/ChakrabortySS22}, our algorithm uses a weighted picking sequence to decide the order following which agents pick their favourite items.
	Thus our algorithm works under the ordinal setting in which we only know the ranking of each agent over the items (instead of the actual costs).
	Our analysis is based on a natural continuous interpretation for the picking sequence algorithm in the weighted setting, which was first used by Li et al.~\cite{conf/www/0037L022} to show the existence of WEF1 allocations for IDO instances.
	Moreover, using the continuous interpretation, we reproduce the proof of Chakraborty et al.~\cite{journals/teco/ChakrabortyISZ21} for the existence of WEF1 allocations for goods.
	We also establish the necessary and sufficient conditions under which the picking sequences can guarantee other weighted fairness notions, e.g., WEF$(x,y)$~\cite{conf/aaai/ChakrabortySS22}.
	
	\smallskip
	
	We also consider allocations that are fair and efficient.
	We consider the \emph{bi-valued} instances, in which there are two values $a, b > 0$ and $a \neq b$ such that $c_i(e) \in \{a,b\}$ for all $i\in N$ and $e\in M$, and propose a polynomial time algorithm that computes WEF1 and PO allocations.
	
	\smallskip
	\noindent
	{\bf Result 2} (Theorem~\ref{the:wef1+po-bivalued}) {\bf .}
	{\em For the allocation of chores to weighted agents, there exists a polynomial time algorithm that computes WEF1 and PO allocations for bi-valued instances.}
	\smallskip
	
	The bi-valued instances are considered as an important special case of the fair allocation problem and have been extensively studied for both goods~\cite{conf/sagt/GargM21,journals/tcs/AmanatidisBFHV21} and chores~\cite{conf/atal/EbadianP022,conf/aaai/GargMQ22,conf/ijcai/0002022}.
	Our result generalizes the results of Garg et al.~\cite{conf/aaai/GargMQ22} and Ebadian et al.~\cite{conf/atal/EbadianP022} to the weighted setting.
	Besides bi-valued instances, we also show that WEF1 and PO allocations exist for two agents (in Appendix~\ref{sec:PO-twoagents}).
	
	\smallskip
	
	Finally, we consider the {\em price of fairness} (PoF) that measures the loss in efficiency due to the fairness constraints.
	In particular, we characterize the price of WEF1, which is the ratio between the minimum social cost of WEF1 allocations and that of the unconstrained allocations.
	For the unweighted case, it has been shown that the price of EF1 is unbounded for three or more agents, and is $\frac{5}{4}$ for two agents~\cite{conf/atal/SunCD21}.
	We generalize the result of Sun et al.~\cite{conf/atal/SunCD21} to the weighted setting by showing that the price of WEF1 is $\frac{4+\alpha}{4}$ for two agents, where $\alpha = \frac{\max\{w_1,w_2\}}{\min\{w_1,w_2\}}$ is the ratio between the weights of agents.
	
	\smallskip
	\noindent
	{\bf Result 3} (Theorem~\ref{theorem:pof-wef1}) {\bf .}
	{\em For the allocation of chores to weighted agents, the price of WEF1 is unbounded for three or more agents, and is $\frac{4+\alpha}{4}$ for two agents where $\alpha = \frac{\max\{w_1,w_2\}}{\min\{w_1,w_2\}}$.}

	\subsection{Other Related Works}
	
	In addition to PROP1 and PROPX, MMS~\cite{conf/bqgt/Budish10} and APS~\cite{conf/sigecom/BabaioffEF21} are two other popular relaxations of PROP.
	For agents with general weights, the weighted version of MMS is studied in~\cite{journals/jair/FarhadiGHLPSSY19,conf/ijcai/0001C019}, and the \emph{AnyPrice Share} (APS) fairness is proposed by Babaioff et al.~\cite{conf/sigecom/BabaioffEF21}.
	They show that there always exist $(3/5)$-approximation of APS allocations for goods and $2$-approximation of APS allocations for chores.
	The approximate ratio for chores is recently improved to $1.733$ by Feige and Huang~\cite{journals/corr/abs-2211-13951}.
	Regarding efficiency, it has been shown that WPROP1 and PO allocations always exist for chores~\cite{journals/corr/abs-1907-01766}, and mixture of goods and chores~\cite{journals/orl/AzizMS20}.
	Whether WPROPX and PO allocations always exist remains unknown.
	
	Besides WEF1, a {\em weak} weighted version of EF1 (WWEF1) is introduced by Chakraborty et al.~\cite{conf/aaai/ChakrabortySS22}.
	When agents are unweighted, both WEF1 and WWEF1 reduce to EF1.
	They show that for the allocation of goods, in the weighted setting, maximizing the \emph{weighted nash social welfare} (WNSW) fails to satisfy WEF1, but guarantees WWEF1.
	They further introduce WEF$(x,y)$ that generalizes WEF1, WWEF1; and WPROP$(x,y)$ that generalizes WPROP1.
	Recently, Aziz et al.~\cite{azizbest} and Hoefer et al.~\cite{journals/corr/abs-2209-03908} study the \emph{weighted envy-freeness up to one transfer} (WEF1T), which is equivalent to WEF$(1,1)$, and propose randomized algorithms for the allocation of goods that guarantee ex-ante WEF and ex-post WEF1T (WEF$(1,1)$).

	\section{Preliminaries}\label{sec:preli}
	
	We consider how to fairly allocate a set of $m$ indivisible items (chores) $M$ to a group of $n$ agents $N$, where each agent $i\in N$ has a weight $w_i > 0$ and $\sum_{i\in N} w_i = 1$.
	When $w_i = 1/n$ for all $i\in N$, we call the instance \emph{unweighted}.
	We call a subset of items, e.g. $X \subseteq M$, a \emph{bundle}.
	Each agent $i\in N$ has an additive cost function $c_i : 2^M \to \bR^+\cup \{0\}$ that assigns a cost to every bundle of items.
	For convenience we use $c_i(e)$ to denote $c_i(\{e\})$, the cost of agent $i\in N$ on item $e\in M$, and thus $c_i(X) = \sum_{e\in X} c_i(e)$ for all $X\subseteq M$.
	Without loss of generality (w.l.o.g.), we assume that the cost functions are normalized, i.e., $\forall i\in N$, $c_i(M) = 1$.
	We use $\bw = (w_1,\ldots,w_n)$ and $\bc = (c_1,\ldots, c_n)$ to denote the weights and cost functions of agents, respectively.
	For ease of notation we use $X+e$ and $X-e$ to denote $X\cup \{e\}$ and $X \setminus \{e\}$, respectively, for any $X\subseteq M$ and $e\in M$.
	%
	An allocation $\bX = (X_1,\ldots,X_n)$ is an $n$-partition of the items $M$ such that $X_i \cap X_j = \emptyset$ for all $i \neq j$ and $\cup_{i\in N} X_i = M$, where agent $i$ receives bundle $X_i$.
	Given an instance $I = (N, M, \bw, \bc)$, our goal is to find an allocation $\bX$ that is \emph{fair} to all agents.
	
	We first introduce the {\em weighted envy-freeness} (WEF) for the allocation of chores.
	
	\begin{definition}[WEF]
		An allocation $\bX$ is \emph{weighted envy-free} (WEF) if for any agents $i, j\in N$, 
		\begin{equation*}
			\frac{c_i(X_i)}{w_i} \leq \frac{c_i(X_j)}{w_j}.
		\end{equation*}
	\end{definition}
	
	Note that when the instance is unweighted, the notion of WEF coincides with the envy-freeness (EF) notion.
	Hence WEF allocations are not guaranteed to exist.
	In the following, we study the weighted envy-freeness up to one item (WEF1), a relaxation of WEF.
	
	\begin{definition}[WEF1]
		An allocation $\bX$ is \emph{weighted envy-free up to one item} (WEF1) if for any agents $i, j\in N$, either $X_i = \emptyset$, or there exists an item $e\in X_i$ such that
		\begin{equation*}
			\frac{c_i(X_i-e)}{w_i} \leq \frac{c_i(X_j)}{w_j}.
		\end{equation*}
	\end{definition}
	
	Finally, we define the Pareto optimality (PO) that evaluates the efficiency of allocations.
	
	\begin{definition}[PO]
		An allocation $X'$ \emph{Pareto dominates} another allocation $X$ if $c_i(X'_i) \leq c_i(X_i)$ for all $i\in N$ and the inequality is strict for at least one agent.
		An allocation $X$ is said to be {\em Pareto optimal} (PO) if $X$ is not dominated by any other allocation.
	\end{definition}
	
	In Appendix~\ref{sec:relationships}, we provide some proofs and examples to show the connection between the notion of WEF1 and other fairness notions.
	In contrast to the allocations of goods where WEF1 fails to imply WPROP1~\cite{journals/teco/ChakrabortyISZ21}, we show that every WEF1 allocation is WPROP1 for the allocation of chores.
	On the other hand, we show that WPROP1 allocations fail to guarantee (any approximation of) WEF1.
	We also show that for the allocation of chores, WEF1 gives a $(2-\min_{i\in N} \{w_i\})$-approximation of APS (a fairness criterion introduced by Babaioff et al.~\cite{conf/sigecom/BabaioffEF21} for weighted agents), and the approximation ratio is tight.

	\section{Weighted EF1 for Chores} \label{sec:compute-wef1}
	
	In this section we present a polynomial time algorithm for computing WEF1 allocations.   
	In the unweighted setting, EF1 allocations for chores can be computed by either the \emph{envy-cycle elimination algorithm}~\cite{conf/sigecom/LiptonMMS04,conf/approx/BhaskarSV21} or the \emph{round-robin algorithm}~\cite{journals/aamas/AzizCIW22}.
	Extending the first algorithm to the weighted setting fails for the case of goods, as shown by Chakraborty et al.~\cite{journals/teco/ChakrabortyISZ21}.
	The main reason is that swapping bundles might not help in resolving envy cycles, when agents involved have different weights.
	For the same reason, the envy-cycle elimination cannot be extended straightforwardly to the weighted case for the allocation of chores.
	In contrast, the round-robin algorithm has been extended successfully to the weighted setting and becomes the \emph{weighted picking sequence protocol}~\cite{journals/teco/ChakrabortyISZ21}, which computes WEF1 allocations for goods.
	Similar to round-robin, in the weighted picking sequence protocol the agents take turns to pick their favourite unallocated items.
	In each round the agent $i$ with minimum $|X_i|/w_i$ to chosen to pick an item (break tie by agent index), where $X_i$ is the set of items agent $i$ receives at the moment.
	Note that for the unweighted case, the algorithm degenerates to the round-robin algorithm.
	Unfortunately, via the following simple example, we show that the algorithm fails to compute WEF1 allocations for chores, even for two agents.
	
	\begin{example} \label{example:hard-forward}
		Consider an instance with $n = 2$ agents, with weights $w_1 = 0.3$ and $w_2 = 0.7$, and $m = 3$ items.
		The two agents have the same cost function, and the costs are shown in Table~\ref{tab:hard-for-good-alg}, where $\epsilon>0$ is arbitrarily small.
		Running the weighted picking sequence protocol gives allocation $X_1 = \{e_1\}$ and $X_2 = \{e_2, e_3\}$, which is not WEF1 because even after removing any item $e$ from $X_2$, $c_2(X_2 - e)/w_2$ is still much larger than $c_2(X_1)/w_1$.
		\begin{table}[htbp]
			\centering
			\begin{tabular}{c|c|c|c}
				&  $e_1$ & $e_2$ & $e_3$ \\ \hline
				agent 1   & $\epsilon$ & $1$ & $1$ \\
				agent 2   & $\epsilon$ & $1$ & $1$
			\end{tabular}
			\caption{Instance showing that weighted picking sequence protocol fails for the allocation of chores.}
			\label{tab:hard-for-good-alg}
		\end{table}
	\end{example}
	
	The algorithm fails because agent $2$ picks more items (compared to that of agent $1$) and this happens \emph{after} agent $1$ picks her favourite item.
	This is not a problem in the case of goods because the items are allocated from the most valuable to the least valuable, but will cause severe envy for the allocation of chores.
	Interestingly, we observe that if we allow agent $2$ to pick two items \emph{before} agent $1$, then the resulting allocation will be WEF1.
	This motivates the design of our algorithm, the \emph{Reversed Weighted Picking Sequence} (RWPS) algorithm, which generates a sequence of agents that is the same as in the weighted picking sequence protocol, but then reverses the sequence to decide the picking sequence of agents.
	While the algorithms are very similar, as we will show in the following section, our analysis is very different from that of Chakraborty et al.~\cite{journals/teco/ChakrabortyISZ21}, and is simpler, and is arguably more intuitive.

	\subsection{Reversed Weighted Picking Sequence}
	\label{ssec:RWPS}
	
	In this section, we propose the Reversed Weighted Picking Sequence (RWPS) Algorithm (see Algorithm~\ref{alg:RWPS} for the details) that computes WEF1 allocations for the allocation of chores.
	
	\paragraph{The Algorithm.}
	We maintain a variable $s_i$ for each agent $i\in N$, where initially $s_i = 0$.
	We call $s_i$ the \emph{size} of agent $i$.
	In the first phase, we generate a length-$m$ sequence of agents $\sigma$ as follows: in round-$t$, where $t = 1,2,\ldots,m$, we let $\sigma(t)$ be the agent with minimum $s_i$ (break tie by picking the agent with minimum index), and increase $s_i$ by $1/w_i$.
	In the second phase, we let agents pick items in the order of $(\sigma(m),\sigma(m-1),\ldots,\sigma(1))$.
	In each agent $i = \sigma(t)$'s turn, she picks her favorite unallocated item, i.e., the item with minimum cost among those that are unallocated, under her own cost function.
	We call $(\sigma(1),\sigma(2), \cdots, \sigma(m))$ the \emph{forward sequence} and $(\sigma(m), \sigma(m-1), \cdots, \sigma(1))$ the \emph{reversed sequence} or \emph{picking sequence}.
	Since the sequence has length $m$, in the final allocation all items are allocated.
	
	\begin{algorithm}[htbp]
		\caption{Reversed Weight Picking Sequence Algorithm} \label{alg:RWPS}
		\KwIn{An instance $<M, N, \bw, \bc>$ with additive cost valuations.}
		Initialize $X_i \gets \emptyset$ and $s_i \gets 0$ for all $i \in N$, and $P \gets M$\;
		\For{$t = 1,2,\ldots,m$}{
			let $i^* \gets \arg \min_{i \in N} \{s_i\}$, break tie by agent index\;
			set $\sigma(t) \gets i^*$, and update $s_{i^*} \gets s_{i^*} + 1/w_{i^*}$\;
		}
		\For{$t = m, m-1, \cdots, 1$}{
			let $i \gets \sigma(t)$ and $e^* \gets \arg \min_{e \in P} \{c_i(e)\}$, break tie by item index\;
			update $X_i \gets X_i + e^*$ and $P \gets P - e^*$\;
		}
		\KwOut{$\bX = \{X_1, X_2 ,\cdots, X_n\}$.}
	\end{algorithm}
	
	In this section, we prove the following main result.
	
	\begin{theorem} \label{thm:compute-wef1}
		For the allocation of chores, the reversed weighted picking sequence algorithm (Algorithm~\ref{alg:RWPS}) computes WEF1 allocations in polynomial time.
	\end{theorem}
	
	Note that the only difference between our algorithm and that of Chakraborty et al.~\cite{journals/teco/ChakrabortyISZ21} is that they use the forward sequence as the picking sequence, while we use the backward one.
	For the allocation of goods, they show that the partial allocation up to round-$t$, for any $t=1,2,\ldots,m$, is always WEF1 and their proof is based on mathematical induction.
	However, in the case of chores, we cannot guarantee that the partial allocation is always WEF1, e.g., consider the instance in Example~\ref{example:hard-forward} (in which agent $2$ picks two items at the beginning of the second phase).
	Therefore, instead of using mathematical induction, we take a continuous perspective on the picking sequence, which was first used to analyze the algorithm that computes WEF1 allocations for chores for IDO instances~\cite{conf/www/0037L022}.
	
	\paragraph{Continuous Perspective.}
	Suppose that in the first phase when we decide the forward sequence, in a round $t$ when agent $i$ is chosen, the variable $s_i$ continuously increases at a rate of $1/w_i$ for one unit of time, and variable $s_j$ does not change, for all $j\neq i$.
	Therefore, we can imagine that $s_i: [0,m] \to [0,m]$ is a non-decreasing continuous function, where $s_i(t)$ denotes the value of variable $s_i$ at time $t\in [0,m]$.
	Similarly, in the second phase when we let agents pick items following the reserved sequence, we assume that in round $t = m,m-1,\ldots,1$, agent $i = \sigma(t)$ \emph{consumes} the item she picks continuously at a rate of $1$, in time interval $(t-1,t]$.
	Note that round $m$ is in fact the first round in the second phase and round $1$ is the last round. 
	To avoid confusion, we will only use the round index $t$ to refer to a round.
	
	\begin{example} \label{example:continuous}
		Consider the following instance with $n=2$ agents and $m=5$ items. Let $w_1 = 0.4$ and $w_2 = 0.6$.
		After executing Algorithm~\ref{alg:RWPS}, the forward sequence is $(1,2,2,1,2)$.        
		We can plot the functions $s_1$ and $s_2$ with a continuous domain as follows, where the time interval $(t-1,t]$ corresponds to round $t$ in Algorithm~\ref{alg:RWPS}.
		At time $t=0,1,2,3,4$, the agent with smaller $s_i(t)$ will grow at a rate of $1/w_i$ until time $t+1$.
		When $s_1$ increases, its rate is given by $1/w_1 = 5/2$; when $s_2$ increases, its rate is given by $1/w_2 = 5/3$. 
		
		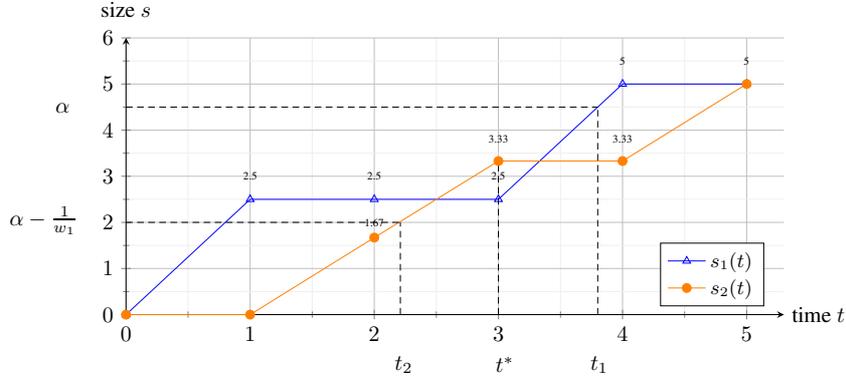
\begin{figure}[!h] 
			\centering 
			\resizebox{.75\textwidth}{!}{%
				\begin{tikzpicture} 
					\begin{axis}[
						axis lines = left,
						xlabel = time $t$,
						ylabel = size $s$,
						every axis x label/.style={at={(current axis.right of origin)},anchor=west},
						every axis y label/.style={at={(current axis.north west)},above=2mm},
						xmin = 0, xmax = 5.3,
						ymin = 0, ymax = 6,
						xtick distance = 1,
						ytick distance = 1,
						grid = both,
						minor tick num = 1,
						major grid style = {lightgray},
						minor grid style = {lightgray!25},
						width = 0.8\textwidth,
						height = 0.4\textwidth,
						legend cell align = {left},
						legend pos = south east
						]
						
						
						\addplot[mark=triangle,blue] plot coordinates {
							(0,0)
							(1,2.5)
							(2,2.5)
							(3,2.5)
							(4,5)
							(5,5)
						};
						\addlegendentry{\small $s_1(t)$}
						
						\addplot[mark=*,orange] plot coordinates { 
							(0,0)
							(1,0)
							(2,1.67)
							(3,3.33)
							(4,3.33)
							(5,5)
						};
						\draw [densely dashed] (0,4.5)--(3.8,4.5); 
						\draw [densely dashed] (3.8,0)--(3.8,4.5); 
						\node at (1,3) {\tiny 2.5} ;
						\node at (2,3) {\tiny 2.5} ;
						\node at (3,3) {\tiny 2.5} ;
						\node at (4,5.5) {\tiny 5} ;
						\node at (5,5.5) {\tiny 5} ;
						\node at (2,2) {\tiny 1.67} ;
						\node at (3,3.8) {\tiny 3.33} ;
						\node at (4,3.8) {\tiny 3.33} ;
						
						\draw [densely dashed] (0,2)--(2.21,2); 
						\draw [densely dashed] (2.21,0)--(2.21,2); 
						\draw [densely dashed] (3,0)--(3,3.33); 
						
						\addlegendentry{\small $s_2(t)$}
					\end{axis}
					
					\node at (7.5,-0.8) {$t_1$};
					\node at (-1,3.3) {$\alpha$};
					\node at (-1.3,1.5) {$\alpha - \frac{1}{w_1}$};
					\node at (4.4,-0.8) {$t_2$};
					\node at (6,-0.8) {$t^*$};
				\end{tikzpicture}
			}
			\vspace{-5pt}
			\caption{The size functions of the two agents, where $s_1(t_1) = \alpha$, $t^* = \lfloor t_1 \rfloor$, and $t_2$ satisfies $s_2(t_2) = \alpha - 1/w_1$.}
			\label{fig:plot-s1-s2}
		\end{figure}
		
		Suppose that after the second phase, $X_1 = \{e_1, e_2\}$, where $e_1$ is chosen by agent $1$ in round $t = 1$ and $e_2$ is chosen in round $t = 4$; $X_2 = \{e'_1, e'_2, e'_3\}$, where $e'_1$, $e'_2$ and $e'_3$ are chosen by agent $2$ in round $t = 2$, $t=3$ and $t=5$, respectively.
		Note that since the picking sequence is the reversed sequence, we have $c_1(e_1) \geq c_1(e_2)$ and $c_2(e'_1)\geq c_2(e'_2)\geq c_2(e'_3)$.
		We can equivalently represent the allocation with continuous time domain as shown in Figure~\ref{fig:continouns-allocation}, where each rectangle represents an item, and has width $1$ and height $1/w_i$, if the item is chosen by agent $i$.
		
		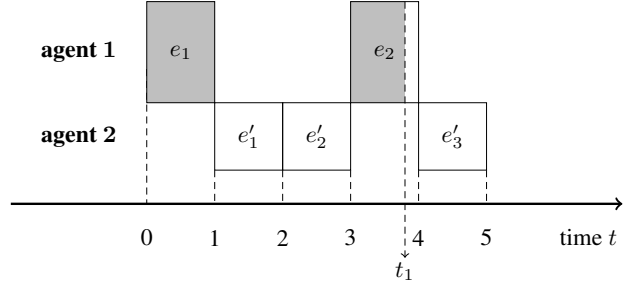
\begin{figure}[htbp]
			\begin{center}
			\resizebox{.55\textwidth}{!}{
				\begin{tikzpicture}
					\node at (-1,1.75) {\textbf{agent} $\mathbf{1}$};
					\node at (-1,0.5) {\textbf{agent} $\mathbf{2}$};
					\draw [densely dashed] (0,1.5)--(0,-0.5); 
					\node at (0, -1) {0};
					\draw [densely dashed] (1,0)--(1,-0.5); 
					\node at (1, -1) {1};
					\draw [densely dashed] (2,0)--(2,-0.5); 
					\node at (2, -1) {2};
					\draw [densely dashed] (3,0)--(3,-0.5); 
					\node at (3, -1) {3};
					\draw [densely dashed] (4,0)--(4,-0.5); 
					\node at (4, -1) {4};
					\draw [densely dashed] (5,0)--(5,-0.5); 
					\node at (5, -1) {5};
					\node at (6.5, -1) {time $t$};
					\draw [->][line width = 1pt] (-2,-0.5)--(7,-0.5);
					\filldraw [fill = gray!50] (0,1) rectangle (1,2.5);
					\node at (0.5,1.75) {$e_1$};
					\draw (1,0) rectangle (2,1);
					\node at (1.5,0.5) {$e'_1$};
					\draw (2,0) rectangle (3,1);
					\node at (2.5,0.5) {$e'_2$};
					\filldraw [fill = gray!50, draw = none] (3,1) rectangle (3.8,2.5);
					\draw (3,1) rectangle (4,2.5);
					\node at (3.5,1.75) {$e_2$};
					\draw (4,0) rectangle (5,1);
					\node at (4.5,0.5) {$e'_3$};
					\draw [densely dashed, ->] (3.8,2.5)--(3.8,-1.3);
					\node at (3.8,-1.5) {$t_1$};
				\end{tikzpicture}
		} 
			\end{center}
			\vspace{-15pt}
			\caption{Illustration of the continuous perspective of the allocation, where the size of the shadow area is $s_1(t_1)$.}
			\label{fig:continouns-allocation}
		\end{figure}
		
	\end{example}
	
	Under the continuous perspective on the algorithm and the allocation, we prove Theorem~\ref{thm:compute-wef1}.
	
	\begin{proofof}{Theorem~\ref{thm:compute-wef1}}
		Fix any two agents $i,j \in N$, we show that agent $i$ is WEF1 towards $j$.
		Suppose that in the final allocation $X_i = \{e_1, e_2, \ldots, e_k\}$, where the items are ordered in increasing order of the index of rounds in which they are chosen.
		Therefore, we have
		\begin{equation*}
			c_i(e_1) \geq c_i(e_2) \geq \cdots \geq c_i(e_k).
		\end{equation*}
		
		Similarly we define $X_j = \{e'_1, e'_2, \ldots, e'_{k'}\}$.
		In the following we show that $\frac{c_i(X_i - e_1)}{w_i} \leq \frac{c_i(X_j)}{w_j}$.
		
		\smallskip
		
		Now we take a continuous perspective and observe how $\frac{c_i(X_i)}{w_i}$ changes when $t$ increases continuously from $0$ to $m$.
		Recall that $s_i(t)$ represents the size of $X_i$ at time $t$ divided by $w_i$, and $s_i(m) = k/w_i$.
		When $s_i(t)$ increases from $(z-1)/w_i$ to $z/w_i$, where $z\in [k]$, item $e_z$ is being consumed by agent $i$ continuously.
		Let $\rho: (0,k/w_i] \rightarrow \bR^+$ be a continuous function such that $\rho(\alpha)$ represents the cost of the item agent $i$ is consuming when $s_i(t)$ reaches $\alpha$.
		In particular, we have
		\begin{equation*}
			\rho(\alpha) = c_i(e_z), \quad \text{ for } \alpha\in \left(\frac{z-1}{w_i}, \frac{z}{w_i} \right], \text{ where } z\in \{1,2,\cdots,k\}.
		\end{equation*}
		
		By definition, $\rho$ is a non-increasing function.
		
		Similarly, we define $\rho': (0,k'/w_j] \rightarrow \bR^+$ be a continuous function such that $\rho'(\alpha)$ represents the cost of the item agent $j$ is consuming, under the cost function of agent $i$, when $s_j(t)$ reaches $\alpha$.
		Hence we have
		\begin{equation*}
			\rho'(\alpha) = c_i(e'_z), \quad \text{ for } \alpha\in \left(\frac{z-1}{w_j}, \frac{z}{w_j} \right], \text{ where } z\in \{1,2,\cdots,k'\}.
		\end{equation*}
		
		By definition of $\rho$ and $\rho'$, we have
		\begin{equation*}
			\frac{c_i(X_i - e_1)}{w_i} = \int_{\frac{1}{w_i}}^{\frac{k}{w_i}} \rho(\alpha) d\alpha, \quad \text{and} \quad
			\frac{c_i(X_j)}{w_j} = \int_{0}^{\frac{k'}{w_j}} \rho'(\alpha) d\alpha.
		\end{equation*}
		
		Next we establish two useful lemmas to show that $\frac{c_i(X_i - e_1)}{w_i} \leq \frac{c_i(X_j)}{w_j}$.
		
		\begin{lemma} \label{lemma:size-of-i-j}
			We have $(k-1)/w_i \leq k'/w_j$.
		\end{lemma}
		\begin{proof}
			Recall that each round when $s_i(t)$ increases, its value grows by an amount of $1/w_i$.
			Consider the moment in time $t'$ when $s_i(t)$ starts to grow from $(k-1)/w_i$, i.e., $t'= \max \{t : s_i(t) = (k-1)/w_i\}$.
			At time $t'$, since $s_i(t')$ is chosen to grow, we must have that $s_i(t') \leq s_j(t')$.
			Therefore we have
			\begin{equation*}
				k'/w_j = s_j(m) \geq s_j(t') \geq s_i(t') = (k-1)/w_i,
			\end{equation*}
			as claimed.
		\end{proof}
		
		\begin{lemma} \label{lemma:density-of-i-j}
			For all $\alpha \in \left(\frac{1}{w_i},\frac{k}{w_i}\right)$, we have $\rho(\alpha) \leq \rho'(\alpha-\frac{1}{w_i})$.
		\end{lemma}
		\begin{proof}
			Fix any $\alpha$ and suppose that $\rho(\alpha) = c_i(e_z)$, i.e., $\alpha \in \left(\frac{z-1}{w_i},\frac{z}{w_i}\right]$.
			Let $t_1$ be the minimum\footnote{Recall that when $\alpha = z/w_i$, there can be multiple values of $t$ for which $s_i(t) = \alpha$.} such that $s_i(t_1) = \alpha$.
			Let $t^*$ be the maximum integer that is smaller than $t_1$.
			By definition, from time $t^*$ to $t^* + 1$, $s_i(t)$ grows from $(z-1)/w_i$ to $z/w_i$ and $t_1 \in (t^*, t^*+1]$.
			Let $t_2$ be minimum time for which $s_j(t_2) = \alpha - 1/w_i$.
			By definition we have $\rho'(\alpha - 1/w_i) = c_i(e'_x)$, where $e'_x$ is the item agent $j$ is consuming at time $t_2$.
			See Figure~\ref{fig:plot-s1-s2} for an example illustrating the definition of $t_1$, $t^*$ and $t_2$, with $i=1$ and $j=2$.
			
			Since at time $t^*$, $s_i$ is chosen to grow, we have
			\begin{equation*}
				s_i(t^*) = (z-1)/w_i \leq s_j(t^*).
			\end{equation*}
			
			Since $\alpha \in \left(\frac{z-1}{w_i},\frac{z}{w_i}\right]$, we have $\alpha - 1/w_i \leq (z-1)/w_i$.
			Recall that $t_2$ is the minimum such that $s_j(t_2) = \alpha - 1/w_i$.
			Since $s_j(t^*) \geq (z-1)/w_i \geq \alpha - 1/w_i$ and $s_j(t)$ is non-decreasing, we have $t_2 \leq t^*$.
			Since $t_1 \in (t^*, t^*+1]$, we have $t_2 \leq t^* < t_1$.
			In other words, in the second phase of Algorithm~\ref{alg:RWPS}, the event that ``agent $i$ includes item $e_z$ in to $X_i$'' happens strictly earlier\footnote{Recall that agents pick items following the reversed sequence, and at most one agent is consuming item, at any point in time.} than the event that ``agent $j$ includes item $e'_x$ into $X_j$''.
			Since agent $i$ picks item $e_z$ when $e'_x$ is still unallocated, we have $c_i(e_z) \leq c_i(e'_x)$, which implies $\rho(\alpha) = c_i(e_z) \leq c_i(e'_x) = \rho'(\alpha-{1}/{w_i})$ and completes the proof.
		\end{proof}
		
		Given Lemma~\ref{lemma:size-of-i-j} and~\ref{lemma:density-of-i-j}, we have
		\begin{equation*}
			\frac{c_i(X_i\setminus\{e_1\})}{w_i} 
			= \int_{\frac{1}{w_i}}^{\frac{k}{w_i}} \rho(\alpha) d\alpha \le \int_{\frac{1}{w_i}}^{\frac{k}{w_i}} \rho'\left(\alpha-\frac{1}{w_i}\right) d\alpha 
			= \int_{0}^{\frac{k-1}{w_i}} \rho'(\alpha) d\alpha \le \int_{0}^{\frac{k'}{w_j}} \rho'(\alpha) d\alpha
			= \frac{c_i(X_j)}{w_j},
		\end{equation*}
		where the first inequality follows from Lemma~\ref{lemma:density-of-i-j} and the second inequality follows from Lemma~\ref{lemma:size-of-i-j}.
		Hence agent $i$ is WEF1 towards agent $j$.
		Since agents $i$ and $j$ are chosen arbitrarily, the allocation is WEF1.
		It is straightforward that Algorithm~\ref{alg:RWPS} runs in $O(mn)$ time, which finishes the proof.
	\end{proofof}
	
	\paragraph{Remark: Allocation of Goods.}
	Note that our continuous perspective can also be applied to analyze the picking sequences for the case of goods~\cite{azizbest,journals/teco/ChakrabortyISZ21,journals/ai/ChakrabortySS21,journals/corr/abs-2209-03908}.
	We take the weighted picking sequence protocol~\cite{journals/teco/ChakrabortyISZ21} as an example and provide in Appendix~\ref{ssec:wef1-goods} an alternative and simpler proof using the continuous perspective to show that the returned allocation is WEF1.

	\subsection{Analysis of General Picking Sequences}
	\label{ssec:general-sequences}
	
	The RWPS algorithm we have introduced falls into the class of picking sequence algorithms.
	\begin{itemize}
		\item In the first phase, the algorithm decides a sequence of agents, which depends on the weights of agents, but is independent of the cost functions.
		\item In the second phase, the algorithm lets agents take turns to pick their favourite unallocated item, following the picking sequence decided in the first phase.
	\end{itemize}
	
	
	In the RWPS algorithm, a forward sequence $\sigma = (\sigma(1),\ldots,\sigma(m))$ is chosen in the first phase, while the reversed sequence $(\sigma(m),\ldots,\sigma(1))$ is used as the picking sequence.
	We refer to this class of algorithms\footnote{In fact, this class is equivalent to the general picking sequence algorithm since we can simply define the reversed one as the sequence. However, for convenience and consistency of our analysis, we let $\sigma$ be the sequence decided in the first phase.} the reversed picking sequence algorithm, and refer to $\sigma = (\sigma(1),\ldots,\sigma(m))$ as the forward sequence, $(\sigma(m),\ldots,\sigma(1))$ as the reversed sequence (which is the actual picking sequence).
	Recall that the size
	\begin{equation}
		s_i(t) = \frac{|\{ t':\sigma(t')=i, t' \leq t\}|}{w_i}
		\label{equation:s-i-t}
	\end{equation}
	of agent $i$ measures the weighted number of appearances of agent $i$ in the sequence up to time $t$, for all $t\in \{0,1,\ldots,m\}$.
	We have shown that if for all $t\in \{1,2,\ldots,m\}$, 
	\begin{equation*}
		\sigma(t) = \arg\min_{i\in N} \{ s_i(t-1) \},
	\end{equation*}
	then the resulting allocation is WEF1.
	In the following, we consider other picking sequences, and establish the conditions under which the resulting allocation satisfies other fairness requirements.
	Note that we can also use the continuous perspective we have introduced, e.g., agents continuously consumes items following the reverse sequence, to analyze the algorithms.
	
	We consider the fairness notions of {\em weakly weighted envy-freeness up to one item} (WWEF1) and {\em weighted envy-freeness up to one transfer} (WEF1T) that are first proposed for the allocation of goods, by Chakraborty et al.~\cite{journals/teco/ChakrabortyISZ21} and Aziz et al.~\cite{azizbest}, respectively.
	In the following, we extend these notions to the allocation of chores.
	
	\begin{definition}[WWEF1]
		An allocation $\bX$ is \emph{weakly weighted envy-free up to one item} (WWEF1) if for any agents $i, j\in N$, there exists an item $e\in X_i$ such that
		\begin{equation*}
			\frac{c_i(X_i-e)}{w_i} \leq \frac{c_i(X_j)}{w_j} \quad \text{or} \quad \frac{c_i(X_i)}{w_i} \leq \frac{c_i(X_j+e)}{w_j}.
		\end{equation*}
	\end{definition}
	
	\begin{definition}[WEF1T]
		An allocation $\bX$ is \emph{weighted envy-free up to one transfer} (WEF1T) if for any agents $i, j\in N$, there exists an item $e\in X_i$ such that
		\begin{equation*}
			\frac{c_i(X_i-e)}{w_i} \leq \frac{c_i(X_j+e)}{w_j}.
		\end{equation*}
	\end{definition}
	
	The above relaxations of weighted envy-freeness, together with WEF1, can be unified under the fairness notion of WEF$(x,y)$, which is first proposed for the allocation of goods by Chakraborty et al.~\cite{conf/aaai/ChakrabortySS22}.
	In the following, we extend it to the allocation of chores.
	
	\begin{definition}[WEF$(x,y)$]
		For any $x,y \in \left [ 0,1 \right ]$, an allocation $\bX$ is \emph{weighted envy-free up to $(x,y)$} (WEF$(x,y)$), if for any agents $i, j\in N$, there exists an item $e\in X_i$ such that
		\begin{equation*}
			\frac{c_i(X_i)-x\cdot c_i(e)}{w_i} \leq \frac{c_i(X_j)+y\cdot c_i(e)}{w_j} 
		\end{equation*}
	\end{definition}
	
	By definition, WEF1 is equivalent to WEF$(1,0)$; WEF1T is equivalent to WEF$(1,1)$; an allocation is WWEF1 if and only if for any agents $i,j,\in N$, agent $i$ is WEF$(1,0)$ or WEF$(0,1)$ towards agent $j$.
	In the following, we establish the conditions under which the allocation returned by the picking sequence algorithm guarantees WEF$(x,y)$.
	
	\begin{theorem}\label{the:wef(x,y)}
		A reversed picking sequence algorithm computes WEF$(x,y)$ allocations if and only if for any $t\in \{ 1,2,\ldots,m \}$ and any agents $i, j\in N$, we have
		\begin{equation*}
			s_i(t) - \frac{x}{w_i} \leq s_j(t) + \frac{y}{w_j},
		\end{equation*}
		where $s_i(t)$ is defined as in~\eqref{equation:s-i-t}.
	\end{theorem}
	\begin{proof}
		We first show that the property is sufficient for ensuring WEF$(x,y)$.
		Fix any agent $i,j\in N$ and let $X_i$ and $X_j$ be the bundles that agent $i$ and $j$ receive in the final allocation, respectively.
		Similar to Theorem~\ref{thm:compute-wef1}, let $X_i = \{e_1, e_2,\dots, e_k\}$ and $X_j = \{e'_1, e'_2,\dots, e'_k\}$, where the items are ordered in increasing order of the index of rounds in which they are chosen. Therefore, we have $c_i(e_1) \geq c_i(e_2) \geq \cdots \geq c_i(e_k)$.
		In the following, we show that
		\begin{equation*}
			\frac{c_i(X_i)-x\cdot c_i(e_1)}{w_i} \leq \frac{c_i(X_j)+y\cdot c_i(e_1)}{w_j}.
		\end{equation*}
		
		Let $\rho: (0,k/w_i] \rightarrow \bR^+$ be a continuous function such that $\rho(\alpha)$ represents the cost of the item agent $i$ is consuming, when $s_i(t)$ reaches $\alpha$.
		In particular, we have
		\begin{equation*}
			\rho(\alpha) = c_i(e_z), \quad \text{ for } \alpha\in \left(\frac{z-1}{w_i}, \frac{z}{w_i} \right], \text{ where } z\in \{1,2,\ldots,k\}.
		\end{equation*}
		
		Similarly, we define $\rho': (0,k'/w_j] \rightarrow \bR^+$ be a continuous function such that $\rho'(\alpha)$ represents the cost of the item agent $j$ is consuming, under the cost function of agent $i$, when $s_j(t)$ reaches $\alpha$.
		Hence we have
		\begin{equation*}
			\rho'(\alpha) = c_i(e'_z), \quad \text{ for } \alpha\in \left(\frac{z-1}{w_j}, \frac{z}{w_j} \right], \text{ where } z\in \{1,2,\ldots,k'\}.
		\end{equation*}
		
		By definition of $\rho$ and $\rho'$, we have
		\begin{equation*}
			\frac{c_i(X_i) -x\cdot c_i(e_1)}{w_i} = \int_{\frac{x}{w_i}}^{\frac{k}{w_i}} \rho(\alpha) d\alpha, \quad \text{and} \quad
			\frac{c_i(X_j)+y\cdot c_i(e_1)}{w_j} = \int_{0}^{\frac{k'}{w_j}} \rho'(\alpha) d\alpha + \frac{y}{w_j}\cdot c_i(e_1).
		\end{equation*}
		
		Using the condition given in the theorem at $t=m$, we obtain the following immediately.
		
		\begin{claim} \label{claim:size-of-i-j}
			We have $\frac{k}{w_i}- \frac{x}{w_i}-\frac{y}{w_j} \le \frac{k'}{w_j}$.
		\end{claim}
		
		Next, we establish a claim that is very similar to Lemma~\ref{lemma:density-of-i-j}.
		
		\begin{claim}\label{claim:density-of-i-j}
			For all $\alpha \in \left(\frac{x}{w_i}+\frac{y}{w_j},\frac{k}{w_i}\right)$, we have  $\rho (\alpha) \le \rho' \left(\alpha-\frac{x}{w_i}-\frac{y}{w_j}\right)$.
		\end{claim}
		\begin{proof}			
			Fix any $\alpha$ and suppose that $\rho(\alpha) = c_i(e_z)$, i.e., $\alpha \in \left(\frac{z-1}{w_i},\frac{z}{w_i}\right]$.
			Let $t_1$ be the minimum such that $s_i(t_1) = \alpha$. 
			Using the condition given in the theorem at $t = t_1$, we have
			\begin{equation*}
				s_i(t_1) - \frac{x}{w_i}-\frac{y}{w_j} = \alpha - \frac{x}{w_i}-\frac{y}{w_j} \leq s_j(t_1).
			\end{equation*}
			
			Let $t_2$ be the minimum such that $s_j(t_2) = \alpha - x/w_i - y/w_j$.
			By definition we have $\rho'(\alpha - x/w_i - y/w_j) = c_i(e'_p)$, where $e'_p$ is the item agent $j$ is consuming at time $t_2$.			
			Since $s_j(t_2) = \alpha - \frac{x}{w_i} - \frac{y}{w_j} \leq s_j(t_1)$ and $s_j(t)$ is non-decreasing, we have $t_2 \leq t_1$.
			In other words, in the second phase of picking sequence algorithm, the event that ``agent $i$ includes item $e_z$ into $X_i$'' happens strictly earlier than the event that ``agent $j$ includes item $e'_p$ into $X_j$''.
			Since agent $i$ picks item $e_z$ when $e'_p$ is still unallocated, we have $c_i(e_z) \leq c_i(e'_p)$, which implies $\rho(\alpha) = c_i(e_z) \leq c_i(e'_p) = \rho'(\alpha-x/w_i-y/w_j)$ and finishes the proof.
		\end{proof}
		
		Combining Claim \ref{claim:size-of-i-j} and \ref{claim:density-of-i-j}, we have
		\begin{equation*}
			\begin{aligned}
				& \frac{c_i(X_i)-x\cdot c_i(e_1)}{w_i} 
				=  \int_{\frac{x}{w_i}}^{\frac{x}{w_i}+\frac{y}{w_j}} \rho(\alpha) d \alpha + \int_{\frac{x}{w_i}+\frac{y}{w_j}}^{\frac{k}{w_i}} \rho(\alpha) d \alpha \\
				\le & \frac{y}{w_j}\cdot c_i(e_1)+ \int_{\frac{x}{w_i}+\frac{y}{w_j}}^{\frac{k}{w_i}} \rho' \left( \alpha-\frac{x}{w_i}-\frac{y}{x_j} \right) d \alpha
				= \frac{y}{w_j}\cdot c_i(e_1) + \int_{0}^{\frac{k}{w_i}-\frac{x}{w_i}-\frac{y}{w_j}} \rho'(\alpha) d \alpha \\
				\le & \frac{y}{w_j}\cdot c_i(e_1) + \int_{0}^{\frac{k'}{w_j}} \rho'(\alpha) d \alpha 
				= \frac{c_i(X_j)+y\cdot c_i(e_1)}{w_j}.
			\end{aligned}
		\end{equation*}
		where the first and second inequalities hold due to Claim \ref{claim:density-of-i-j} and Claim \ref{claim:size-of-i-j}, respectively. 
		
		\smallskip
		
		Finally, we show that the condition is necessary for ensuring WEF$(x,y)$.
		Assume otherwise, and let $t \in \{1,2,\ldots,m\}$ be such that $s_i(t) - \frac{x}{w_i} > s_j(t) + \frac{y}{w_j}$ for some $i,j\in N$.
		Consider an instance with $t$ items having cost $1$ to all agents and $m - t$ items having cost $0$ to all agents.
		Following the reversed picking sequence, we have
		\begin{equation*}
			\frac{c_i(X_i)-x\cdot c_i(e)}{w_i} = s_i(t) - \frac{x}{w_i}, \quad \text{ and } \quad
			\frac{c(X_j)+y\cdot c_i(e)}{w_j} = s_j(t)  +  \frac{y}{w_j},
		\end{equation*}
		which implies that agent $i$ is not WEF$(x,y)$ towards agent $j$, and is a contradiction.
	\end{proof}

	From Theorem~\ref{the:wef(x,y)}, we have the following corollaries.
	
	\begin{corollary}\label{coro:wef1}
		A reversed picking sequence algorithm computes WEF1 allocations if and only if for any $t\in \{ 1,2,\ldots,m \}$ and any pair of agents $i,j\in N$, we have $s_i(t) - \frac{1}{w_i} \leq s_j(t)$.
	\end{corollary}
	
	Consequently, to ensure WEF1, RWPS is basically the only algorithm one can use.
	Any other algorithm that ensures WEF1 can be regarded as RWPS with a different tie-breaking rule.
	
	\begin{corollary}
		A reversed picking sequence algorithm computes WEF1T (WEF$(1,1)$) allocations if and only if for any $t\in \{ 1,2,\ldots,m \}$ and any pair of agents $i,j\in N$, we have $s_i(t) - \frac{1}{w_i} \leq s_j(t) + \frac{1}{w_j}$.
	\end{corollary}
	
	\begin{corollary}
		A reversed picking sequence algorithm computes WWEF1 allocations if and only if for any $t\in \{ 1,2,\ldots,m \}$ and any pair of agents $i,j\in N$, we have $s_i(t) \leq s_j(t) + \frac{1}{ \min\{w_i, w_j\} }$.
	\end{corollary}
	
	Furthermore, we show that we can compute WEF$(x,y)$ allocations for any $x+y \geq 1$, by providing an algorithm similar to Algorithm~\ref{alg:RWPS} (we defer the algorithm and analysis to Appendix~\ref{sec:wefxy}).

	%

	\section{WEF1 and PO for Bivalued Instances}
	
	In this section, we focus on the computation of allocations that are fair and efficient, and explore the existence of WEF1 and PO allocations.
	Garg et al.~\cite{conf/aaai/GargMQ22} and Ebadian et al.~\cite{conf/atal/EbadianP022} show that EF1 and PO allocations exist for the bi-valued instances (see below for the definition) when agents have equal weights.
	In this section we prove a more general result that WEF1 and PO allocations always exist and can be computed efficiently for the bi-valued instances, using a similar proof framework.
	
	We first give the definition of bi-valued instances.
	
	\begin{definition}[Bi-valued Instances]
		An instance is called \emph{bi-valued} if there exist constants $a,b \geq 0$ such that for any agent $i\in N$ and item $e\in M$ we have $c_i(e) \in \{a,b\}$.
	\end{definition}
	
	For non-zero $a,b$ with $a \neq b$, we can scale the cost functions so that $c_i(e) = \{1, k\}$ for some $k > 1$.
	If $c_i(e) = k$ we call item $e$ \emph{large} to agent $i$; otherwise we call it \emph{small} to $i$.
	Moreover, if there exists an agent $i\in N$ such that $c_i(e) = k$ for all $e\in M$, we can rescale the costs so that $c_i(e) = 1$ for all $e\in M$.
	Hence we can assume w.l.o.g. that for all $i\in N$, there exists at least one item $e\in M$ such that $c_i(e) = 1$.
	
	In this section we prove the following main result.
	
	\begin{theorem}\label{the:wef1+po-bivalued}
		There exists an algorithm that computes a WEF1 and PO allocation for any given bi-valued instance in polynomial time.
	\end{theorem}
	
	We first classify the items into two groups depending on their costs as follows.
	
	\begin{definition}[Item Groups] \label{def:large-small-item}
		We call item $e\in M$ a \emph{consistently large} item if for all $i\in N, c_i(e) = k$.
		Let $M^+$ include all consistently large items, and $M^-$ contain the other items:
		\begin{equation*}
			M^+ = \{e\in M : \forall i\in N, c_i(e)=k\}, \quad  M^- = \{e\in M : \exists i\in N, c_i(e)=1\}.
		\end{equation*}
	\end{definition}
	
	
	\paragraph{Fisher Market.}
	In the Fisher market, there is a \emph{price} vector $\bp$ that assigns each chore $e\in M$ a price $p(e) > 0$.
	For any subset $X_i \subseteq M$, let $p(X_i) = \sum_{e\in X_i} p(e)$.
	Given the price vector $\bp$, we define the \emph{pain-per-buck} ratio $\alpha_{i,e}$ of agent $i$ for chore $e$ to be $\alpha_{i,e} = c_i(e)/p(e)$, and the \emph{minimum pain-per-buck} (MPB) ratio $\alpha_i$ of agent $i$ to be $\alpha_i = \min_{e\in M} \{ \alpha_{i,e} \}$.
	For each agent $i$, we define $\MPB_i = \{e\in M: \alpha_{i,e} = \alpha_i\}$, and we call each item $e \in \MPB_i$ an MPB item of agent $i$.
	An allocation $\bX$ with price $\bp$ forms a (Fisher market) equilibrium $(\bX, \bp)$ if each agent only receives her MPB chores, i.e. $X_i \subseteq \MPB_i$ for any $i\in N$.
	
	\begin{definition}[pWEF1]
		An equilibrium $(\bX, \bp)$ is called \emph{price weighted envy-free up to one item} (pWEF1) if for any $i,j\in N$, there exists an item $e\in X_i$ such that
		\begin{equation*}
			\frac{p(X_i-e)}{w_i} \leq \frac{p(X_j)}{w_j}.
		\end{equation*}
	\end{definition}
	
	Throughout this section, we call $p(X_i)/w_i$ the (weighted) spending of agent $i$.
	For convenience of notation, given an equilibrium $(\bX, \bp)$, we use	 $\hat{p}_i$ to denote the (weighted) spending of agent $i$ after removing the item with maximum price, i.e.,
	\begin{equation*}
		\hat{p}_i = \min_{e\in X_i} \left\{ \frac{p(X_i-e)}{w_i} \right\}.
	\end{equation*}
	
	In the following, we say that agent $i$ \emph{strongly envies} agent $j$ if $\hat{p}_i > \frac{p(X_j)}{w_j}$.
	Note that the equilibrium is pWEF1 if and only if no agent strongly envies another agent.
	
	\begin{lemma}
		If an equilibrium $(\bX, \bp)$ is pWEF1, the allocation $\bX$ is WEF1 and PO.
	\end{lemma}
	\begin{proof}
		We first show that the allocation $\bX$ is Pareto optimal.
		If the allocation $\bX$ with price $\bp$ is an equilibrium, then for any agent $i\in N$, any item $e\in X_i$ and any $j\neq i$ we have 
		\begin{equation*}
			\frac{c_i(e)}{\alpha_i} = \frac{c_i(e)}{\alpha_{i,e}} = p(e) = \frac{c_j(e)}{\alpha_{j,e}} \le \frac{c_j(e)}{\alpha_{j}}.
		\end{equation*}
		
		Thus the allocation minimizes the objective $\sum_{i\in N} \frac{c_i(X_i)}{\alpha_i}$. 
		Any Pareto improvement would strictly decrease this objective, which leads to a contradiction.
		So the allocation $\bX$ is PO.
		
		Next, we show that the allocation $\bX$ is WEF1.
		Since the equilibrium $(\bX, \bp)$ is pWEF1, for any agents $i,j\in N$, there exists an item $e\in X_i$ such that $\frac{p(X_i-e)}{w_i} \leq \frac{p(X_j)}{w_j}$. 
		Note that both agents $i$ and $j$ only receive items holding the MPB ratio since $(\bX, \bp)$ is an equilibrium.
		Then we have
		\begin{equation*}
			\frac{c_i(X_i-e)}{w_i} = \alpha_i \cdot \frac{p(X_i-e)}{w_i} \leq \alpha_i \cdot \frac{p(X_j)}{w_j} \leq \frac{c_i(X_j)}{w_j},
		\end{equation*}
		where last equality holds since the pain-per-buck ratio of agent $i$ on any item is at least $\alpha_i$.
	\end{proof}
	
	\begin{definition}[Big and Least Spenders] \label{def:big-least-spenders}
		Given an equilibrium $(\bX, \bp)$, an agent $b\in N$ is called a big spender if $b = \arg\max_{i\in N} \left\{\hat{p}_i\right\}$; an agent $l$ is called a least spender if $l = \arg\min_{i\in N} \left\{ \frac{p(X_i)}{w_i} \right\}$\footnote{Throughout the whole paper, when selecting a big spender we break tie by picking the agent with smallest index. The same principle applies to the selection of least spender. Therefore in the rest of the paper we call $b$ (resp. $l$) ``the'' big (resp. least) spender when this tie breaking rule is applied.}.
	\end{definition}
	
	We show that if the big spender does not strongly envy the least spender, then the allocation is pWEF1.
	
	\begin{lemma}{\label{lemma:pwef1}}
		If an equilibrium $(\bX, \bp)$ holds that the big spender $b$ does not strongly envy the least spender $l$, then the equilibrium is pWEF1.
	\end{lemma}
	\begin{proof}
		For any agent $i,j\in N$, we show that agent $i$ does not strongly envy agent $j$.
		Note that $b$ is the big spender and $l$ is the least spender.
		From the definitions we have 
		\begin{equation*}
			\min_{e\in X_i} \left\{ \frac{p(X_i-e)}{w_i} \right\} = \hat{p}_i \leq \hat{p}_b = \min_{e\in X_b} \left\{ \frac{p(X_b-e)}{w_b} \right\} \quad \text{and} \quad
			\frac{p(X_l)}{w_l} \leq \frac{p(X_j)}{w_j}.
		\end{equation*}
		
		Recall that the big spender $b$ does not strongly envy the least spender $l$, i.e.
		\begin{equation*}
			\min_{e\in X_b} \left\{ \frac{p(X_b-e)}{w_b} \right\} \leq \frac{p(X_l)}{w_l}.
		\end{equation*}
		
		Hence we have $ \min_{e\in X_i} \left\{ \frac{p(X_i-e)}{w_i} \right\} \leq \frac{p(X_j)}{w_j}$, and thus $(\bX, \bp)$ is pWEF1.
	\end{proof}
	
	Given the above lemma, to see if $(\bX, \bp)$ is pWEF1, it suffices to consider the envy from the big spender $b$ to the least spender $l$.
	If $b$ strongly envies $l$, then we try to reallocate some item from $X_b$ to $X_l$, and possibly update the price of some items, while ensuring that the resulting allocation and price form a new equilibrium.
	We show that such reallocations are always possible, and by polynomially many reallocations, we can eliminate the envy from the big spender to the least spender. 
	To begin with, we first compute an initial equilibrium $(\bX^0, \bp^0)$, based on which we partition the agents into different groups.

	\subsection{Initial Equilibrium and Agent Groups}
	
	In the following, we give an algorithm (Algorithm~\ref{alg:leximin}) that computes an initial equilibrium $(\bX, \bp)$ and agent groups $\{N_r\}_{r\in [R]}$ with the some useful properties (see Lemma~\ref{lemma:leximin}).
	
	\paragraph{The Initial Price.}
	We set the price vector $\bp$ as $p(e) = \min_{i\in N} \{c_i(e)\}$.
	Therefore we have $p(e) = 1$ for all $e\in M^-$ and $p(e) = k$ for all $e\in M^+$.
	
	\paragraph{The Allocation.}
	We first compute an allocation $\bX$ that minimizes the social cost: for each $e\in M^-$ we allocate it to an arbitrary agent $i$ with $c_i(e) = 1$; items in $M^+$ are allocated arbitrarily.
	Then we construct a directed graph $G_X$ based on $\bX$ as follows.
	For each pair of $i,j\in N$, if there exists an item $e\in X_i$ such that $e\in \MPB_j$ then we create an MPB edge from $j$ to $i$.
	While there exists a path from an agent $j$ to another agent $i$ such that $\hat{p}_i > \frac{p(X_j)}{w_j}$, we implement a sequence of item transfers backward along the path.
	When there is no path of such type, we finish the computation of the initial allocation $\bX$.
	
	\paragraph{The Agent Groups.}
	We select the big spender $b_1$ and recognize the agents that can reach $b_1$ via MPB paths.
	We add these agents to the first agent group $N_1$, together with agent $b_1$.
	After identifying group $N_1$, we repeat the above procedure by picking the big spender $b_2$ among the remaining agents $N\setminus N_1$ and let $N_2$ contain $b_2$ and the agents that can reach $b_2$ via MPB paths.
	Recursively, we partition agents into groups $(N_1, N_2, \cdots, N_R)$.
	We call $N_1$ the \emph{highest} group and $N_R$ the \emph{lowest}.
	By construction, each group $N_r$ has a \emph{representative} agent $b_r$, to which every agent in $N_r \setminus \{ b_r \}$ has an MPB path.
	
	\begin{algorithm}[htbp]
		\caption{Computation of Initial Price, Allocation and Agent Groups} \label{alg:leximin}
		\KwIn{A bi-valued instance $<M, N, \bw, \bc>$}
		initialize $X_i \gets \emptyset$ for all $i\in N$, $P \gets M$ \;
		set $p(e) =1$ for all $e\in M^-$ and $p(e) = k$ for all $e\in M^+$ \;
		
		\tcp{Phase 1: Computation of the Initial Allocation}
		
		compute a social cost minimizing allocation $\bX$ \;
		construct a direct graph $G_X = (N,E)$: add an MPB edge from $j$ to $i$ if $X_i \cap \MPB_j \neq \emptyset$ \;
		\While{there exists a path $i_k \to \cdots \to i_0$ such that $\hat{p}_{i_0} > \frac{p(X_{i_k})}{w_{i_k}}$}{
			\For{$l = 1,2,\cdots,k$}{
				pick an item $e\in X_{i_{l-1}} \cap \MPB_{i_l}$, break tie by picking the item with maximum price \;
				update $X_{i_{l-1}} \gets X_{i_{l-1}} - e$ and $X_{i_l} \gets X_{i_l} + e$ \;
			}
		}
		
		\tcp{Phase 2: Computation of Agent Groups}
		
		initialize $R \gets 0, N' \gets N$ \;
		\While{$N' \neq \emptyset$}{
			let $b \gets \argmax_{i\in N'} \{ \hat{p}_i \}$, break tie by picking the agent with smallest index  \;
			update $R\gets R+1$ and let $N_R \gets \{b\} \cup \{ i\in N' : \text{exists an MPB path from $i$ to $b$ in $G_X$} \}$ \;
			$N' \gets N' \setminus N_R$ \;
		}
		\KwOut{$\bX = \{X_1, \cdots, X_n\}$, $\bp$, $\{N_r\}_{r\in[R]} = \{N_1, \cdots, N_R\}$.}
	\end{algorithm} 
	
	We say that a group $N_r$ is pWEF1 if all agents in $N_r$ are pWEF1 towards each other.
	
	\begin{lemma}\label{lemma:leximin}
		Algorithm~\ref{alg:leximin} returns an allocation $\bX$ with price $\bp$ and agent groups $\{N_r\}_{r\in[R]}$ with the following properties
		\begin{enumerate}
			\item $(\bX, \bp)$ is an equilibrium, and $\alpha_i = 1$ for all $i\in N$.
			\item For all $i,j\in N$, if $j$ is from a group lower than $i$ (i.e., there exist $r < r'$ such that $i\in N_r, j\in N_{r'}$) then for all $e\in X_i$ we have $c_j(e) = k$.
			\item All consistently large items are allocated to the lowest group: $M^+ \subseteq \bigcup_{i\in N_R} X_i$.
			\item For all $r\in [R]$, the agent group $N_r$ is pWEF1.
		\end{enumerate}
	\end{lemma}
	\begin{proof}
		For property 1: by the way we set the initial price, we have $\alpha_i = \min_{e\in M} \left\{ \frac{c_i(e)}{p(e)} \right\} = 1$ for all $i\in N$.
		Since we start from the social cost minimizing allocation (in which $X_i \subseteq \MPB_i$ for all $i\in N$) and reallocate an item $e$ to agent $i$ only if $e\in \MPB_i$, we can ensure that $X_i \subseteq \MPB_i$ for all $i\in N$ in the final allocation.
		Hence $(\bX, \bp)$ is an equilibrium.
		
		For property 2: suppose that there exists $e\in X_i$ such that $c_j(e) =1$.
		Then we have $e\in \MPB_j$ and there is an MPB edge from $j$ to $i$.		
		Let $b_r$ be the representative agent of group $N_r$.
		By the construction of group $N_r$, $i$ can reach $b_r$ via an MPB path.
		Hence $j$ can also reach $b_r$ via an MPB path,	which is a contradiction because it implies that $j$ should be included into $N_r$.
		
		For property 3: following a similar argument, suppose that there exists an item $e \in M^+$ that is allocated to an agent $i\in N_r$ with $r < R$.
		Then every agent $j\in N_R$ has an MPB edge to $i$ since $e\in \MPB_j$.
		Thus $j$ can reach the representative agent $b_r$ of $N_r$, which leads to a contradiction since $j\notin N_r$.
		
		For property 4: recall that every agent $i\in N_r$ can reach the representative $b_r$ via an MPB path, and $\hat{p}_{b_r} \geq \hat{p}_i$.
		Moreover, we have $\hat{p}_{b_r} \leq \frac{p(X_i)}{w_i}$ because otherwise the path from $i$ to $b_r$ should have been resolved in the first phase of Algorithm~\ref{alg:leximin} when $\bX$ is computed.
		
		Hence for all $i, j\in N_r$, we have $\hat{p}_i \leq \hat{p}_{b_r} \leq \frac{p(X_j)}{w_j}$, which implies that $i$ is pWEF1 toward $j$.        
	\end{proof}
	
	\begin{lemma}
		Algorithm~\ref{alg:leximin} runs in polynomial time.
	\end{lemma}
	\begin{proof}
        Observe that (1) the computation of the initial price $\bp$ takes $O(m)$ time; (2) computing the social cost minimizing allocation takes $O(nm)$ time; (3) in Phase 2, given the initial allocation, computing the agent groups takes $O(nm)$ time because each group can be computed by identifying a connected component of $G_X$. 
		Therefore, to prove that the algorithm runs in polynomial time, it suffices to argue that the while loops in lines 5 - 8 of Algorithm~\ref{alg:leximin} finish in polynomial time.
		Observe that computing $G_X$ takes $O(nm)$ time and resolving an MPB path takes $O(m)$ time.
		In the following, we show that by carefully choosing the MPB paths to resolve, the while loops break after $O(knm)$ rounds.
		
		Recall that in each while loop, we identify a path $i_k \to \cdots \to i_0$ with $\hat{p}_{i_0} > {p(X_{i_k})}/{w_{i_k}}$, and resolve the path by transferring an MPB item from agent $i_{l-1}$ to $i_{l}$, for all $l=1,2,\ldots,k$.
		We call $i_k$ the start-agent of the path and $i_0$ the end-agent.
        When there exist multiple such paths, we choose the one maximizing $\hat{p}_{i_0}$.
        Observe that such a path can be identified in $O(nm)$ time.
        We refer to a while loop as a round, and index the rounds by $t=1,2,\ldots$.
        We use $X^t_i$ to denote the bundle of agent $i$ at the beginning of round $t$.
        Likewise we define $\hat{p}^t_i$ and the path $i_k^t \to \cdots \to i_0^t$ at the beginning of round $t$.
        
        \begin{claim}\label{claim:decreasing-end-agent}
            We have $\hat{p}^0_{i_0^0} \geq \hat{p}^1_{i_0^1} \geq \cdots \geq  \hat{p}^t_{i_0^t}$.
        \end{claim}
        \begin{proof}
        Consider a round $t$ in which path $i_k^t \to \cdots \to i_0^t$ is chosen.
        We show that for all agent $i_l^t$, where $l \neq k$, her spending would not increase after the item transfers in this round, i.e. $p(X^t_{i_l^t}) \geq p(X^{t+1}_{i_l^t})$.
        Assume otherwise, then it must be that $i_l^t$ receives some item $e$ with $p(e) = k$ from $i_{l-1}^t$ and $i_l^t$ transfers an item $e'$ with $p(e') = 1$ to agent $i_{l+1}^t$.
        However, this is impossible because the consistently large item $e$ is an MPB item to all agents, which should be transferred to agent $i_{l+1}^t$ as we break tie by choosing the item with maximum price.
        Hence in each round, only the spending of the start-agent would increase.
        This implies $\hat{p}^{t+1}_{i_0^{t+1}} \leq \hat{p}^t_{i_0^t}$ because at the beginning of round $t+1$, compared to round $t$, the only agent whose spending is increased is $i_k^t$, and its current spending is $\hat{p}^{t+1}_{i_k^t} \leq {p(X^t_{i_k^t})}/{w_{i_k^t}} < \hat{p}^t_{i_0^t}$.
        \end{proof}
            
        We further show that if an agent is identified as an end agent twice, then her spending (up to the removal of one item) is strictly smaller at its second appearance.
        
        \begin{claim}
            If agent $i$ is the end-agent in both rounds $t_1$ and $t_2$, where $t_1 < t_2$, then we have $\hat{p}^{t_1}_i > \hat{p}^{t_2}_i$.
        \end{claim}
        \begin{proof}
            Assume otherwise, i.e., $\hat{p}^{t_1}_i = \hat{p}^{t_2}_i$, then $i$ must receive some item as a start-agent in some round between $t_1$ and $t_2$.
            Let $t < t_2$ be the last round in which $i$ is a start-agent.
            Let $j$ be the corresponding end-agent.
            We have $\hat{p}^{t}_j > \frac{p(X^{t}_i)}{w_i}$.
            Since $t$ is the last round before $t_2$ in which $i$'s spending increases, we have $\hat{p}^{t_2}_i \leq \frac{p(X^t_i)}{w_i} < \hat{p}^{t}_j$, which is a contradiction because from Claim~\ref{claim:decreasing-end-agent}, we have $\hat{p}^{t_2}_i = \hat{p}^{t_1}_i \geq \hat{p}^{t}_j$.
        \end{proof}
        
        Note that for each agent $i$, the value of $\hat{p}_i$ is at most $\frac{km-1}{w_i}$.
        As argued above, each agent can be identified as an end-agent at most $km$ times because each of its appearances decreases $\hat{p}_i$ by at least $\frac{1}{w_i}$.
        Hence the total number of rounds is at most $knm$.
        In summary, Algorithm~\ref{alg:leximin} finishes in $O(k n^2 m^2)$ time.
	\end{proof}
	
	\subsection{The Allocation Algorithm and the Invariants}\label{sec:wef1+po-bi}
	
	In this section we present an algorithm that starts from the initial equilibrium (denoted by $(\bX^0, \bp^0)$), and constructs a pWEF1 equilibrium $(\bX, \bp)$ by a sequence of item reallocations and price raises.
	The algorithm proceeds in rounds.
	In each round we identify the big spender $b$ and the least spender $l$.
	\begin{itemize}
		\item If $b$ is pWEF1 towards $l$ then the algorithm terminates and outputs the allocation.
		By Lemma~\ref{lemma:pwef1}, we can guarantee that the output equilibrium is pWEF1.
		\item Otherwise we reallocate an item from $b$ to $l$ following an MPB path, during which we may raise the price of all items owned by agents from the group containing $b$ by a factor of $k$.
	\end{itemize}
	
	Before we present the full details of the algorithm, we remark that throughout the whole process, the following invariants are always maintained. 
	For convenience of notation, we use $N_{\leq i}$ to denote $\bigcup_{j\leq i} N_j$. 
	Likewise, we define $N_{<i}, N_{\geq i}$ and $N_{>i}$.
	
	\begin{invariant}[Equilibrium Invariant] \label{invariant:equilibrium}
		At any point of time we have $(\bX, \bp)$ being an equilibrium.
	\end{invariant}
	
	\begin{invariant}[pWEF1 Invariant] \label{invariant:pwef1}
		For all $r\in [R]$, agent group $N_r$ is pWEF1 in equilibrium $(\bX, \bp)$.
	\end{invariant}
	
	\begin{invariant}[Raised Group Invariant] \label{invariant:raised-group}
		There exists $r^* \in [R]$ such that all groups $N_1,\ldots,N_{r^* - 1}$ are raised exactly once; all groups $N_{r^*},\ldots,N_R$ are not raised.
		Moreover
		\begin{itemize}
			\item for all $i\in N_{< r^*}$ we have $\alpha_i = 1/k$ and $X_i \subseteq X^0_i$, i.e., agent $i$ did not receive any new item;
			\item for all $i\in N_{\geq r^*}$ we have $\alpha_i = 1$ and $X^0_i \subseteq X_i$, i.e., agent $i$ did not lose any item.
		\end{itemize}
	\end{invariant}
	
	Note that both invariants hold true at the beginning of the algorithm when $(\bX, \bp) = (\bX^0, \bp^0)$ and all groups are not raised (e.g., $r^* = 1$), by Lemma~\ref{lemma:leximin}.
	Moreover, Invariant~\ref{invariant:raised-group} implies that the last group $N_R$ is never raised because $R\geq r^*$ for all $r^* \in [R]$, which further implies that $M^+\subseteq \bigcup_{i\in N_R} X_i$ because unraised agents did not lose any item.
	
	\begin{corollary} \label{corollary:property-last-group}
		Group $N_R$ is not raised, and $M^+\subseteq \bigcup_{i\in N_R} X_i$.
	\end{corollary}
	
	Next, we present the full details of the algorithm (see Algorithm~\ref{alg:ef1andpo} for the pseudo-code).
	
	\paragraph{The Full Algorithm.}
	We first call Algorithm~\ref{alg:leximin}, which returns the initial equilibrium $(\bX, \bp)$, and modify the equilibrium until it becomes pWEF1.
	For analysis purpose, we use $(\bX^0, \bp^0)$ to denote this initial equilibrium.
	In each round (of the while loop), we identify the big spender $b$ and the least spender $l$.
	If $b$ is pWEF1 toward $l$ then we output $(\bX, \bp)$ and terminate (Lemma~\ref{lemma:pwef1} ensures that the equilibrium is pWEF1).
	Otherwise we check whether $l$ is raised.
	If $l$ is not raised, then we raise the group containing $b$ if this group has not been raised, and then allocate an item from $b$ to $l$ (we can show that all items in $X_b$ are MPB items to $l$).
	If $l$ is raised, then we can show (see Lemma~\ref{lemma:existence-b-i-l}) that $b$ must be raised already, and there exists an unraised agent $i$ that received item from $l$ in some past round.
	Then we return the item we have reallocated from $l$ to $i$ back to $X_l$, and reallocate an item from $b$ to $l$.
	
	\begin{algorithm}[H]
		\caption{Find a WEF1 and PO allocation} \label{alg:ef1andpo}
		\SetKw{Break}{Break}
		\KwIn{A bi-valued instance $<M, N, \bw, \bc>$}
		let $(\bX, \bp, \{N_r\}_{r\in [R]})$ be returned by Algorithm~\ref{alg:leximin} ; \ \tcp{Denote the equilibrium by $(\bX^0,\bp^0)$}
		initialize the set of unraised agents $U\gets N$ \;
		\While{ \textbf{True}}{
			let $b \gets \arg\max_{i\in N} \{ \hat{p}_i \}$ be the big spender \;
			let $l \gets \arg\min_{i\in N} \left\{ \frac{p(X_i)}{w_i} \right\}$ be the least spender \;
			\If{$\hat{p}_b \leq \frac{p(X_l)}{w_l}$}{
				\textbf{output} $(\bX, \bp)$ and \textbf{terminate}.
			}
			suppose that $b \in N_r$ and $l\in N_{r'}$ ; \qquad\qquad \tcp{we have $r \neq r'$ due to Invariant~\ref{invariant:pwef1}}
			\uIf {$l \in U$}{
				\If{$b \in U$}{
					raise the price of all chores in $\bigcup_{i\in N_r} X_i$ by a factor of $k$ \;
					$U \gets U \setminus N_r$ \;
				}
				pick an arbitrary $e\in X_b$ ; \qquad\qquad\qquad \tcp{we have $X_b\subseteq \MPB_l$, see Lemma~\ref{lemma:X-b-in-MPB-l}}
				update $X_b \gets X_b - e$, $X_l \gets X_l + e$ \;}
			\Else {
				we have $b\in N\setminus U$ and $\exists i \in U$ with $X_i \cap X^0_l \neq \emptyset$ ; \qquad \qquad \qquad \tcp{see Lemma~\ref{lemma:existence-b-i-l}}
				pick any $e_1\in X_b \cap \MPB_i$ and $e_2\in X_i\cap X^0_l$ \;
				update $X_b \gets X_b - e_1$, $X_l \gets X_l + e_2$ and $X_i \gets X_i + e_1 - e_2$ \;
			}
		}
	\end{algorithm} 
	
	\subsection{Properties of the Algorithm}
	
	\paragraph{Notations.}
	In the following, we refer to a while loop as a round, and index the rounds by $t = 0,1,2,\ldots$.
	We use $\bX^t$ and $\bp^t$ to denote the allocation and the price at the beginning of round $t$\footnote{Note that the notations are consistent with the notation for the initial equilibrium $(\bX^0, \bp^0)$ returned by Algorithm~\ref{alg:leximin}.}.
    We use $b^t$ and $l^t$ to denote the big spender and the least spender, respectively, we identify at the beginning of round $t$.
	Note that $\bX^{t+1}$ and $\bp^{t+1}$ are the allocation and the price at the end of round $t$.
	
	\smallskip
	
	In addition to the invariants we have introduced, we introduce a few more invariants that involve multiple rounds.
	It is trivial to check that these invariants hold true at the beginning of the first round (when $t=0$).
	
	\begin{invariant}[Least Spending Invariant] \label{invariant:least-spending}
		The spending of the least spender across different rounds is non-decreasing: $\frac{p^0(X^0_{l^0})}{w_{l^0}} \leq \frac{p^1(X^1_{l^1})}{w_{l^1}} \leq \cdots \leq \frac{p^t(X^t_{l^t})}{w_{l^t}}$.
	\end{invariant}
	
	\begin{invariant} [Big Spender Invariant] \label{invariant:big-spender}
		The big spender $b^t$ identified at the beginning of round $t$ has never been identified as a least spender in rounds $1,2,\ldots,t$.
		Moreover, if $b$ has not been raised, i.e., $b \in N_r$ for some $r \geq r^*$, then this property holds true for all $i\in N_r$.
	\end{invariant}
	
	\begin{invariant}[Price Invariant] \label{invariant:price}
		In each round, the spending of agent $i$ does not change, unless $i$ is raised during this round or $i$ is the big or least spender in this round.
		Moreover, for all $e\in M$, $p(e)\in \{1,k\}$.
	\end{invariant}
	
	\smallskip
	
	
	In the following, we fix some round $t$, and assume that at the beginning of round $t$, the Equilibrium Invariant (Invariants~\ref{invariant:equilibrium}), the pWEF1 Invariant (Invariant~\ref{invariant:pwef1}), the Raised Group Invariant (Invariant~\ref{invariant:raised-group}), the Least Spending Invariant (Invariant~\ref{invariant:least-spending}), the Big Spender Invariant (Invariant~\ref{invariant:big-spender}) and the Price Invariant (Invariant~\ref{invariant:price}) are maintained.
	Recall that all invariants hold true at the beginning of round $0$.
	Moreover, we assume the algorithm is well defined thus far.
	We assume that $b^t$ strongly envies $l^t$ (otherwise the algorithm terminates).
	We show several properties of the algorithm, which will be used to show that the invariants are maintained at the end of round $t$.
	When the context is clear, we drop the superscript $t$ for ease of notation.
	We first show a few properties regarding the big spender $b$ and least spender $l$ in round $t$, which show that the algorithm is well defined.
	Let $r^*$ be defined as in Invariant~\ref{invariant:raised-group}, i.e., all agents in $N_{< r^*}$ are raised exactly once and all agents in $N_{\geq r^*}$ are not raised.
	Using Invariant~\ref{invariant:big-spender}, we prove that if the big spender $b$ has not been raised, then $b$ must be in the unraised group with smallest index.
	This would be helpful in showing that Invariant~\ref{invariant:raised-group} is maintained at the end of round $t$: when we raise the group containing $b$, we need to ensure that the raised group is $N_{r^*}$.
	
	\begin{lemma} \label{lemma:b-in-N-r*}
		If $b \in U$, then we have $b\in N_{r^*}$.
	\end{lemma}
	\begin{proof}
		By Invariant~\ref{invariant:raised-group}, for all unraised agent $i\in U$, we have $X^0_i \subseteq X_i$.
		That is, agent $i$ did not lose any item.
		Moreover, by the design of the algorithm, $X^0_i$ is a proper subset of $X_i$ if and only if $i$ was identified as a least spender in some round before $t$.
		
		Suppose that $b\in N_r$ for some $r\geq r^* +1$.
		By Invariant~\ref{invariant:big-spender}, agent $b$ has never been identified as a  least spender, which (together with $b\in U$) implies that $X_b = X^0_b$ and $\hat{p}_b = \hat{p}^0_b$.
		Let $b' = \arg\max_{i\in N_{r^*}} \{ \hat{p}^0_i \}$.
		Note that $\hat{p}_{b'} \geq \hat{p}^0_{b'}$ because $b'$ did not lose any item.
		Since $b$ is the big spender, we have
		\begin{equation*}
			\hat{p}^0_b = \hat{p}_b \geq \hat{p}_{b'} \geq  \hat{p}^0_{b'} \geq  \hat{p}^0_{b},
		\end{equation*}
		where the last inequality holds due to the principle for computing the initial agent groups.
		Therefore all inequalities should hold with equality, which leads to a contradiction because $b'$ should be the big spender in round $t$ (recall that we use the same rule for breaking tie, i.e., by selecting the agent with minimum index, every time we select a big spender).
	\end{proof}
	
	The above lemma implies the following almost straightforwardly.
	
	\begin{lemma} \label{lemma:X-b-in-MPB-l}
		If $l \in U$, then we have $X_b \subseteq \MPB_l$ in line 13.
	\end{lemma}
	\begin{proof}
		Suppose that $l\in N_{r'}$ and $b\in N_r$.
		Since $b$ strongly envies $l$, by Invariant~\ref{invariant:pwef1}, we know that $b$ and $l$ are from different groups. 
		Additionally, by Lemma~\ref{lemma:b-in-N-r*}, we have $r \leq r^* < r'$. 
		By Invariant~\ref{invariant:big-spender}, we have $X_b \subseteq X^0_b$.
		Then by the property of agent groups (see Lemma~\ref{lemma:leximin}), we have $c_l(e) = k$ for all $e\in X_b$.
		Since in line 13, agent $b$ has been raised already, for all $e\in X_b$ we have $p(e) = k$, which implies that $e\in \MPB_l$ because $\alpha_l = 1$ (by Invariant~\ref{invariant:raised-group}).
	\end{proof}
	
	The next lemma shows that line 16 of Algorithm~\ref{alg:ef1andpo} is well defined.
	
	\begin{lemma} \label{lemma:existence-b-i-l}
		If $l$ has been raised, then $b$ has also been raised, and there exists an unraised agent $i\in U$ that received some item from $l$ in some round before $t$.
	\end{lemma}
	\begin{proof}
		We first show that $X_l$ is a proper subset of $X^0_l$.
		By Invariant~\ref{invariant:raised-group}, since $l$ is raised, we have $X_l \subseteq X^0_l$.
		Suppose $X_l = X^0_l$, then either (1) $l$ has never been identified as a big spender (and thus did not lose any item); or (2) $l$ was identified as a big spender and lost some item, but was identified as a (raised) least spender later and retrieved the lost item.
		
		For case (1), let $t' < t$ be the round in which $l$ is raised.
		By assumption $l$ is in the same group as $b^{t'}$ and thus $b^{t'}$ is pWEF1 towards $l$ in $(\bX^{t'},\bp^{t'})$, by Invariant~\ref{invariant:pwef1}.
		Thus we have 
		\begin{equation*}
			\frac{p(X_l)}{w_l} \geq \frac{k\cdot p^{t'}(X^{t'}_{l})}{w_l} \geq k\cdot \hat{p}^{t'}_{b^{t'}}  \geq k\cdot \hat{p}^{t'}_{b} \geq \hat{p}_b,
		\end{equation*}
		where the first inequality holds because $l$ did not lose any item, and $l$ is raised during round $t'$; the second inequality holds because $b^{t'}$ is pWEF1 toward $l$ in $(\bX^{t'},\bp^{t'})$; the third inequality holds because $b^{t'}$ is the big spender at the beginning of round $t'$; the last inequality holds because $b$ did not receive any item before round $t$, by Invariant~\ref{invariant:big-spender}.
		Then we have a contradiction because $b$ strongly envies $l$ (in $(\bX,\bp)$).
		
		For case (2), let $t_1$ be the round during which $l$ is raised and $t_2 \geq t_1$ be the first round in which $l$ is identified as a big spender.
		By definition of $t_2$ and Invariant~\ref{invariant:big-spender}, $l$ did no lose or receive any item before round $t_2$.
		Therefore we have $X^{t_1}_l = X^0_l = X_l$.
		Since $l$ is raised during round $t_1$, following the same analysis as in case (1) we have
		\begin{equation*}
			\frac{p(X_l)}{w_l} = \frac{k\cdot p^{t_1}(X^{t_1}_{l})}{w_l} \geq k\cdot \hat{p}^{t_1}_{b^{t_1}}  \geq k\cdot \hat{p}^{t_1}_{b} \geq \hat{p}_b,
		\end{equation*}
		which also contradicts with the fact that $b$ strongly envies $l$ (in $(\bX,\bp)$).
		
		Therefore, $l$ must have lost some item (as a big spender) before round $t$ and the item has not been retrieved.
		Let $e\in X^0_l \setminus X_l$ be any such item and suppose $e\in X_i$.
		Then by Invariant~\ref{invariant:raised-group}, $i$ has not been raised when round $t$ begins, because $X_i \not\subseteq X^0_i$.
		It remains to show that $b$ is raised (when round $t$ begins).
		
		Assume that $b$ has not been raised when round $t$ begins.
		Let $t' < t$ be the last round during which agent $l$ loses some item as a big spender. 
		Then we have 
		\begin{equation*}
			\hat{p}^{t'}_l \leq \frac{p(X)}{w_l} < \hat{p}_b \leq \hat{p}^{t'}_b,
		\end{equation*}
		where the first inequality holds because $t'$ is the last round during which $l$ loses some item, the second inequality holds because $b$ strongly envies $l$ in $(\bX,\bp)$, the third inequality holds since we assume that $b$ has never been raised, and did not receive any item (by Invariant~\ref{invariant:big-spender}).
		Therefore we have a contradiction because $b$ should be the big spender in round $t'$.
		Hence $b$ must have been raised when round $t$ begins, and the proof is complete.
	\end{proof}

	\subsection{Maintenance of Invariants}
	
	We show in this section that assuming all the stated invariants hold true when round $t$ begins (after identifying the big and least spenders), all of them will continue to hold at the end of round $t$ (assuming that the algorithm does not terminate during round $t$).
	
	\begin{lemma}
		The Equilibrium Invariant (Invariant~\ref{invariant:equilibrium}) is maintained at the end of round $t$.
	\end{lemma}
	\begin{proof}
		To prove that $(\bX^{t+1}, \bp^{t+1})$ is an equilibrium, we need to argue that 
		\begin{itemize}
			\item[(1)] raising the price of a group does not cause $X_i\nsubseteq \MPB_i$ for any $i\in N$.
			\item[(2)] all reallocations happened during round $t$ follow MPB edges.
		\end{itemize}
		
		For the first property, note that by Invariant~\ref{invariant:big-spender} and~\ref{invariant:raised-group}, if agent $i\in N_{r^*}$ is raised during this round, then we have $X_i = X^0_i$ and $p(e) = 1$ for all $e\in X_i$.
		\begin{itemize}
			\item For all $j \in N_{<r^*}$, since $\alpha_j = 1/k$, raising group $N_{r^*}$ does not affect $X_j \subseteq \MPB_j$.
			\item For all $j \in N_{r^*}$, $\alpha_j$ decreases from $1$ to $1/k$ but for all $e\in X_j$ we have $\alpha_{j,e} = 1/k$, which implies $X_j \subseteq \MPB_j$.
			\item For all $j\in N_{>r^*}$, by Lemma~\ref{lemma:leximin}, we have $c_j(e) = k$ for all item $e$ whose price is raised (from $1$ to $k$). Hence $\alpha_j$ remains $1$ and we still have $X_j \subseteq \MPB_j$.
		\end{itemize}

		For the second property, we first consider the case when $l$ is not raised.
		By Lemma~\ref{lemma:X-b-in-MPB-l}, the item allocated to $l$ is in $\MPB_l$.
		Now suppose that $l$ is raised.
		Then by Lemma~\ref{lemma:existence-b-i-l}, $b$ is also raised, and the algorithm reallocates an item in $e_2\in X_i\cap X^0_l$ to $l$ and an item $e_1$ from $b$ to $i$.
		By Invariant~\ref{invariant:raised-group}, we have $\alpha_b = \alpha_l = 1/k$ and $\alpha_i = 1$.
		When $e_2$ was allocated from $l$ to $i$, $p(e_2)$ has already been raised (from $1$ to $k$).
		Thus $\alpha_{l,e_2} = \frac{c_l(e_2)}{p(e_2)} = 1/k$, which implies that $e_2\in \MPB_l$.
		Since $b$ is in a raised group and $i$ is not raised, we have $c_i(e_1) = k$.
		Furthermore, since $p(e_1) = k$, we have $\alpha_{i,e_1} = 1$, which implies that $e_1\in \MPB_i$.
	\end{proof}
	
	\begin{lemma}
		The pWEF1 Invariant (Invariant~\ref{invariant:pwef1}) is maintained at the end of round $t$.
	\end{lemma}
	\begin{proof}
		By Invariant~\ref{invariant:price}, during round $t$, only the price of the big spender $b$, the least spender $l$ and agents whose price is raised during this round would change.
		Moreover, raising the price of a whole group does not affect the pWEF1 property within this group.
		Hence it suffices to consider the case when no agent is raised during this round and argue that the group $N_r$ containing $b$ and the group $N_{r'}$ containing $l$ remain pWEF1 when round $t$ ends.
		By Invariant~\ref{invariant:pwef1}, both groups $N_r$ and $N_r'$ are pWEF1 when round $t$ begins.
		\begin{itemize}
			\item For $N_r$, since only the spending of $b$ decreases, it suffices to argue that no agent $i\in N_r\setminus\{b\}$ envies $b$ when round $t$ ends.
			Note that by Invariant~\ref{invariant:price}, the spending of $i$ does not change.
			Since the spending of $b$ at the end of this round is at least $\hat{p}_b$ (since $b$ loses one item in this round), and $\hat{p}_i \leq \hat{p}_b$ (since $b$ is the big spender), we conclude that $i$ is pWEF1 towards $b$ when round $t$ ends.
			
			\item For $N_{r'}$, since only the spending of $l$ increases, it suffices to argue that agent $l$ does not envy any other agent $i\in N_{r'}$ when round $t$ ends, which trivially holds because agent $l$ has the minimum spending when round $t$ begins, and $l$ receives only one item during this round.
		\end{itemize}
		
		Hence the equilibrium $(\bX^{t+1}, \bp^{t+1})$ is also pWEF1.
	\end{proof}

	\begin{lemma}
		The Raised Group Invariant (Invariant~\ref{invariant:raised-group}) is maintained at the end of round $t$.
	\end{lemma}
	\begin{proof}
		It suffices to consider the case when the group containing $b$ is raised during this round.
		Recall that by Lemma~\ref{lemma:b-in-N-r*}, we have $b\in N_{r^*}$.
		Hence raising $N_{r^*}$ causes $r^*$ to be increased by one, at the end of round $t$.
		Note that we must have $r^* \leq R-1$ because otherwise ($b\in N_R$) the least spender $l\notin N_R$ is raised, which implies that $b$ is already raised (by Lemma~\ref{lemma:existence-b-i-l}), and is a contradiction.
		Moreover, by Invariant~\ref{invariant:big-spender}, every agent $i\in N_{r^*}$ did not receive any item, which implies $X^{t+1}_i \subseteq X^0_i$.
		Since for all $i\in N_{r^*}$ we have $\alpha_i = 1$ when round $t$ begins, we have $\alpha_i = 1/k$ when round $t$ ends.
		Therefore Invariant~\ref{invariant:raised-group} is maintained.
	\end{proof}

	\begin{lemma}
		The Least Spending Invariant (Invariant~\ref{invariant:least-spending}) is maintained at the end of round $t$.
	\end{lemma}
	\begin{proof}
		To show that $\frac{p^{t+1}(X^{t+1}_{l^{t+1}})}{w_{l^{t+1}}} \geq \frac{p(X_l)}{w_l}$, it suffices to argue that the spending of every agent $i\in N$ is at least $\frac{p(X_l)}{w_l}$. 
		By Invariant~\ref{invariant:price}, only the spending of $b$ decreases.
		Since $b$ loses only one item, her spending at the end of round $t$ is at least $\hat{p}_b$.
		Since $b$ strongly envies $l$, we have $\hat{p}_b\geq \frac{p(X_l)}{w_l}$, which implies that the spending of $b$ when round $t$ ends is at least $\frac{p(X_l)}{w_l}$. 		
	\end{proof}
	
	\begin{lemma}
		The Big Spender Invariant (Invariant~\ref{invariant:big-spender}) is maintained at the end of round $t$.
	\end{lemma}
	\begin{proof}
		Recall that the big and least spenders in round $t+1$ are decided by the equilibrium we compute at the end of round $t$.
		To prove the lemma, we need to show that 
		\begin{itemize}
			\item[(1)] at the beginning of round $t+1$, the big spender $b^{t+1}$ has never been identified as a least spender in rounds $1,2,\ldots,t+1$;
			\item[(2)] if $b^{t+1} \in N_r$ has not been raised (when round $t+1$ begins) then all agent $i\in N_r$ has never been identified as a least spender in rounds $1,2,\ldots,t+1$.
		\end{itemize}
		
		We first consider the case when $b^{t+1}$ has not been raised when round $t+1$ begins.
		Assume the contrary of statement, and let $t' \leq t$ be the last round in which some agent $i\in N_r$ was identified as the least spender, i.e., $i = l^{t'}$.
		Note that $i$ can be $b^{t+1}$.
		By definition, in rounds $t'+1,t'+2,\ldots,t$, agents in $N_r$ did not receive any item; in round $t'$ only agent $i\in N_r$ receives one item.
		\begin{itemize}
			\item If $i\neq b^{t+1}$, we have
			\begin{equation*}
				\hat{p}^{t+1}_{b^{t+1}} \leq \hat{p}^{t'}_{b^{t+1}} \leq \frac{p^{t'}(X^{t'}_i)}{w_i} \leq \frac{p^{t+1}(X^{t+1}_{l^{t+1}})}{w_{l^{t+1}}},
			\end{equation*}
			where the first inequality holds because $b^{t+1}$ did not receive any item and is not raised in rounds $t', t'+1,\ldots, t$; the second inequality holds because agents $i$ and $b^{t+1}$ are in the same group; and the last inequality follows from Invariant~\ref{invariant:least-spending}.
			
			\item Similarly, if $i = b^{t+1}$, we have
			\begin{equation}
				\hat{p}^{t+1}_{b^{t+1}} \leq \hat{p}^{t'+1}_{b^{t+1}} \leq \frac{p^{t'}(X^{t'}_{b^{t+1}})}{w_{b^{t+1}}} \leq \frac{p^{t+1}(X^{t+1}_{l^{t+1}})}{w_{l^{t+1}}}, \label{equation:big-spender-not-as-least}
			\end{equation}
			where the first inequality holds because $b^{t+1}$ did not receive any item and is not raised in rounds $t'+1, t'+2, \ldots, t$; the second inequality holds because agent $b^{t+1}$ loses only one item during round $t'$; and the last inequality follows from Invariant~\ref{invariant:least-spending}.
		\end{itemize}
		
		Therefore we have a contradiction since $b^{t+1}$ strongly envies $l^{t+1}$.
		
		It remains to consider the case that $b^{t+1}$ has been raised when round $t+1$ begins, and prove statement (1).
		Assume otherwise and let $t' \leq t$ be the last round in which $b^{t+1}$ is identified as the least spender.
		If $b^{t+1}$ has been raised when round $t'$ begins, then Inequality~\eqref{equation:big-spender-not-as-least} remains valid, and we have a contradiction.
		If $b^{t+1}$ is not raised when round $t'$ begins, then $b^{t+1}$ can not be the big spender or raised in rounds $t'+1, t'+2, \ldots, t$ because $X^{t'+1}_{b^{t+1}} \nsubseteq X^0_{b^{t+1}}$, by Invariant~\ref{invariant:big-spender} and~\ref{invariant:raised-group}.
		Therefore we also have a contradiction as we assumed that $b^{t+1}$ has been raised when round $t+1$ begins. 
	\end{proof}

	\begin{lemma}
		The Price Invariant (Invariant~\ref{invariant:price}) is maintained at the end of round $t$.
	\end{lemma}
	\begin{proof}
		We first show that $p(e) \in \{1,k\}$ for all $e\in M$ when round $t$ ends.
		Suppose $N_r$ is raised during round $t$.
		By Invariant~\ref{invariant:big-spender}, for all $i\in N_r$, agent $i$ did not receive any item before round $t$, i.e. $X_i \subseteq X^0_i$.
		By Invariant~\ref{invariant:raised-group}, for all $i\in N_r$ we have $\alpha_i = 1$, which implies that for all $e\in X_i \subseteq \MPB_i$, $p(e) = \frac{c_i(e)}{\alpha_i} = 1$.
		Hence raising the price of $e$ does not violate the invariant.
		
		Next, we show that for all agent $i$ that is not raised, nor the big or least spender, the spending of $i$ does not change during round $t$.
		If $l$ is not raised when round $t$ begins, then the algorithm reallocates an item from $b$ to $l$, and the statement trivially holds.
		If $l$ is already raised when round $t$ begins, then it suffices to consider that case when the reallocate happens following path $b \to i \to l$.
		By Lemma~\ref{lemma:existence-b-i-l}, $i$ (which is not raised) receives an item $e_1 \in X_b\cap \MPB_i$ and loses an item $e_2 \in X_i \cap X^0_l$.
		Since every item our algorithm reallocates is raised, we have $p(e_1) = p(e_2) = k$, which implies that the spending of $i$ remains unchanged after updating $X_i \gets X_i + e_1 - e_2$.
	\end{proof}
	
	Given that all invariants are maintained, to prove Theorem~\ref{the:wef1+po-bivalued}, it suffices to argue that the algorithm terminates in a polynomial number of rounds.
	
	\begin{lemma}
		Algorithm~\ref{alg:ef1andpo} terminates in polynomial time.
	\end{lemma}
	\begin{proof}
		It can be verified that each round of Algorithm~\ref{alg:ef1andpo} finishes in $O(m+n)$ time.
		Therefore it suffices to argue that the algorithm terminates after a polynomial number of rounds.	
        We consider the least spenders across different rounds and we focus on the appearances of an arbitrary agent $i$.
        By the design of our algorithm, in every round in which $i$ is the least spender, the size of $X_i$ increases by $1$.
        Moreover, by Invariant~\ref{invariant:big-spender}, agent $i$ would not be a big spender after the first time $i$ is identified as a least spender.
        Hence $i$ can be identified as a least spender at most $m$ times because each of its appearance (as a least spender) increases its bundle size by one.
        Therefore, the total number of rounds is at most $nm$, which implies that the algorithm terminates in $O((m+n)nm)$ time.
	\end{proof}
	
	\section{Price of Fairness} \label{sec:pof}
	
	In this section we focus on the price of fairness, which is introduced in~\cite{journals/ior/BertsimasFT11,journals/mst/CaragiannisKKK12} to capture the efficiency loss due to the fairness constraints.
	As in previous works~\cite{conf/atal/SunCD21,journals/mst/BeiLMS21,conf/www/0037L022}, we assume that the cost functions are normalized, i.e. $c_i(M) = 1$ for all $i\in N$.
	We use $\sc(\bX) = \sum_{i\in N} c_i(X_i)$ to denote the social cost of allocation $\bX$.
	In this section we analyze the price of fairness (PoF) for WEF1.	
	Given an instance $\mathcal{I}$, the PoF with respect to WEF1 is the ratio between the minimum social cost of WEF1 allocations and the (unconstrained) optimal social cost.
	
	\begin{definition}[PoF for WEF1]
		The price of fairness for WEF1 on a given instance $\mathcal{I}$ is defined as
		\begin{equation*}
			\POF(\mathcal{I}) = \min_{\text{WEF1 allocation } \bX} \left\{ \frac{\sc(\bX)}{\opt (\mathcal{I})} \right\},   
		\end{equation*}
		where $\opt(\mathcal{I})$ denotes the unconstrained optimal social cost.
		The price of fairness for WEF1 is defined as the maximum price of fairness  over all instances, e.g., $\POF = \sup_{\mathcal{I}} \{ \POF(\mathcal{I}) \}$.
	\end{definition}
	
	For the unweighted setting, Sun et al.~\cite{conf/atal/SunCD21} have shown that the price of EF1 is unbounded when $n\geq 3$.
	For $n=2$, the price of EF1 is $\frac{5}{4}$ and there exists an algorithm that computes an EF1 allocation $\bX$ achieving a social cost $\sc(\bX) \leq \frac{5}{4}\cdot \opt(\mathcal{I})$ for any given instance $\mathcal{I}$ on two agents.
	In this section, we consider the price of fairness for WEF1 (price of WEF1).
	The result of Sun et al.~\cite{conf/atal/SunCD21} implies that the price of WEF1 is unbounded for $n\geq 3$.
	Thus we consider the case when $n=2$ and show that the price of WEF1 is $(4+\alpha)/4$, where $\alpha = \frac{\max\{w_1, w_2\}}{\min\{w_1, w_2\}}\geq 1$ is the weight ratio between the two agents. 
	Note that the ratio coincides with $\frac{5}{4}$ in the unweighted setting when $\alpha = 1$.
	
	\begin{theorem}\label{theorem:pof-wef1}
		The price of WEF1 is $(4+\alpha)/4$ for instances with two agents and $\frac{\max\{w_1, w_2\}}{\min\{w_1, w_2\}} = \alpha$.
	\end{theorem}
	
	To prove the theorem, we first provide an instance $\mathcal{I}$ for which any WEF1 allocation has social cost at least $\frac{4+\alpha}{4} \cdot \opt(\mathcal{I})$; then we present an algorithm that computes a WEF1 allocation with social cost at most $\frac{4+\alpha}{4} \cdot \opt(\mathcal{I})$, for any given instance $\mathcal{I}$.
	
	\smallskip
	
	Consider the following instance $\mathcal{I}$ with $w_1 = \frac{\alpha}{1+\alpha}$ and $w_2 = \frac{1}{1+\alpha}$, for some $\alpha\geq 1$.
	The cost functions of the two agents are shown in Table~\ref{tab:hard-pof}, where $\epsilon>0$ is arbitrarily small.
	
	\begin{table}[htbp]
		\centering
		\begin{tabular}{c|c|c|c}
			&  $e_1$ & $e_2$ & $e_3$ \\ \hline
			Agent 1   & $0$ & $\cfrac{1}{2}$ & $\cfrac{1}{2}$ \\
			Agent 2   & $\cfrac{\alpha}{\alpha+2}-2\epsilon$ & $\cfrac{1}{\alpha+2}+\epsilon$ & $\cfrac{1}{\alpha+2}+\epsilon$
		\end{tabular}
		\caption{Hard instance for lower bounding the price of WEF1.}
		\label{tab:hard-pof}
	\end{table}
	
	For the given instance the optimal social cost is $\opt(\mathcal{I}) = \frac{2}{\alpha+2} + 2\epsilon$, achieved by the allocation $X_1 = \{e_1\}$, $X_2 = \{e_2, e_3\}$.
	However, the allocation is not WEF1 since for all $e\in X_2$, we have
	\begin{equation*}
		\frac{c_2(X_2 - e)}{w_2} = \frac{\alpha+1}{\alpha+2} + \epsilon\cdot(\alpha+1) > \frac{\alpha+1}{\alpha+2} - 2\epsilon\cdot(1+\frac{1}{\alpha}) = \frac{c_2(X_1)}{w_1}.
	\end{equation*}
	
	Hence, any WEF1 allocation $\bX$ allocates at most one item in $\{e_2, e_3\}$ to agent $2$, leading to
	\begin{equation*}
		sc(\bX) \geq \frac{1}{2} + \frac{1}{\alpha+2} + \epsilon.
	\end{equation*}
	
	Therefore the price of WEF1 for instances with two agents and $\frac{\max\{w_1, w_2\}}{\min\{w_1, w_2\}} = \alpha$ is 
	\begin{equation*}
		\POF \geq \frac{\frac{1}{2} + \frac{1}{\alpha+2} + \epsilon}{\frac{2}{\alpha+2} + 2\epsilon} = \frac{\alpha+4+2\epsilon\cdot(\alpha+2)}{4+4\epsilon\cdot(\alpha+2)} = \frac{4+\alpha}{4} + O(\epsilon).
	\end{equation*}
	
	\smallskip
	
	Next, we present an algorithm that computes WEF1 allocations achieving the price of WEF1.
	
	\begin{lemma}\label{lemma:upperbound-price}
		There exists an algorithm that computes a WEF1 allocation with social cost at most $\frac{4+\alpha}{4}\cdot \opt(\mathcal{I})$, for any given instance $\mathcal{I}$ with two agents having weight ratio $\frac{\max\{w_1, w_2\}}{\min\{w_1, w_2\}} = \alpha$.
	\end{lemma}	
	\begin{proof}
		We consider the weighted adjusted winner algorithm (see Algorithm~\ref{alg:WAW}) for two agents.
		We index the items in non-decreasing order of their cost-ratios $\frac{c_1(e)}{c_2(e)}$, i.e.,
		\begin{equation*}
			\frac{c_1(e_1)}{c_2(e_1)} \le \frac{c_1(e_2)}{c_2(e_2)} \le \dots \le \frac{c_1(e_m)}{c_2(e_m)}.
		\end{equation*}
		
		Let $O_1 = \{ e\in M: c_1(e) < c_2(e) \}$ and $O_2 = \{ e\in M: c_1(e) \geq c_2(e) \}$. 
		Note that the allocation $\mathbf{O} = (O_1, O_2)$ minimizes the social cost, i.e., $\sc(\mathbf{O}) = \opt(\mathcal{I})$.
		If $\mathbf{O}$ is WEF1 then our algorithm terminates and outputs $\mathbf{O}$.
		Otherwise we compute a WEF1 allocation as follows.
		
		For all integer $t$, let $L(t) = \{ e_1,e_2,\ldots,e_t\}$ and $R(t) = \{ e_t, e_{t+1}, \ldots, e_m \}$.
		Note that $L(t)$ and $R(t)$ can be empty, e.g., when $t < 1$ or $t > m$.
		Let $f$ be the maximum index satisfying
		\begin{equation*}
			\frac{c_2(R(f+1))}{w_2} > \frac{c_2(L(f-1))}{w_1}.
		\end{equation*}
		
		We return the allocation $\bX$ with $X_1 = L(f)$ and $X_2 = R(f+1)$.
		
		\begin{algorithm}[htbp]
			\caption{Weighted Adjusted Winner Algorithm} \label{alg:WAW}
			\KwIn{Instance with two agents satisfying $\frac{c_1(e_1)}{c_2(e_1)} \le \frac{c_1(e_2)}{c_2(e_2)} \le \dots \le \frac{c_1(e_m)}{c_2(e_m)}$}
			let $O_1 = \{ e\in M: c_1(e) < c_2(e) \}$ and $O_2 = \{ e\in M: c_1(e) \geq c_2(e) \}$ \;
			\uIf{allocation $(O_1, O_2)$ is WEF1}{
				\KwOut{$\mathbf{O} = \{O_1, O_2\}$.}
			}
			\Else{
				assume w.l.o.g. that $\frac{c_1(O_1)}{w_1} \leq \frac{c_2(O_2)}{w_2}$ (otherwise we reverse the index of agents and items) \;
				find the maximum index $f$ such that $\frac{c_2(R(f+1))}{w_2} > \frac{c_2(L(f-1))}{w_1}$ \;
				$X_1 \gets L(f)$, $X_2 \gets R(f+1)$ \;
				\KwOut{$\bX = \{X_1, X_2\}$.}
			}
		\end{algorithm}
		
		In the following we show that the allocation $\bX$ is WEF1 and $\sc(\bX) \leq \frac{4+\alpha}{4}\cdot \sc(\mathbf{O})$.
		By definition of $f$ we know that agent $2$ is WEF1 towards agent $1$ because
		\begin{equation*}
			\frac{c_2(X_2 - e_{f+1})}{w_2} = \frac{c_2(R(f+2))}{w_2} \leq \frac{c_2(L(f))}{w_1} = \frac{c_2(X_1)}{w_1}.
		\end{equation*}
		
		We also have that agent $1$ is WEF1 towards agent $2$ because
		\begin{align*}
			\frac{c_1(X_1 - e_f)}{w_1} & = \frac{c_1(L(f-1))}{w_1} = \frac{c_1(L(f-1))}{c_2(L(f-1))} \cdot \frac{c_2(L(f-1))}{w_1}  \\
			& < \frac{c_1(L(f-1))}{c_2(L(f-1))} \cdot \frac{c_2(R(f+1))}{w_2} 
			\leq \frac{c_1(X_2)}{c_2(X_2)} \cdot \frac{c_2(R(f+1))}{w_2} = \frac{c_1(X_2)}{w_2},
		\end{align*}
		where the second inequality follows from (recall that the items are sorted in non-decreasing order of $\frac{c_1(e)}{c_2(e)}$)
		\begin{equation*}
			\frac{c_1(L(f-1))}{c_2(L(f-1))} = \frac{\sum_{t=1}^{f-1} c_1(e_t)}{\sum_{t=1}^{f-1} c_2(e_t)} \le \frac{c_1(e_{f})}{c_2(e_{f})} \le \frac{\sum_{t = f+1}^{m} c_1(e_t)}{\sum_{t = f+1}^{m} c_2(e_t)} = \frac{c_1(X_2)}{c_2(X_2)}.
		\end{equation*}
		
		Next, we prove that $\sc(\bX) = c_1(X_1) + c_2(X_2) \leq \frac{4+\alpha}{4}\cdot \sc(\mathbf{O})$.
		
		Recall that $O_1 = \{ e\in M: c_1(e) < c_2(e) \}$ and $O_2 = \{ e\in M: c_1(e) \geq c_2(e) \}$, which implies $c_1(O_1) \leq c_2(O_1)$ and $c_2(O_2) \leq c_1(O_2)$.
		Moreover, we assumed w.l.o.g. that $\frac{c_1(O_1)}{w_1} \leq \frac{c_2(O_2)}{w_2}$.
		For ease of notation in the following we use $A$ and $B$ to denote $c_1(O_1)$ and $c_2(O_2)$, respectively. 
		Then we have $\frac{A}{w_1} \leq \frac{B}{w_2}$ and $\sc(\mathbf{O}) = A + B$.
		Since the cost functions are normalized, we have
		\begin{equation*}
			c_1(O_2) = 1-A \quad \text{and} \quad c_2(O_1) = 1- B.
		\end{equation*}
		
		In the following, we upper bound $\sc(\bX)$ by a function of $A$ and $B$.
		
		\begin{claim} \label{claim:A-B-bound-f-at-right}
			We have $A \leq w_1$, $B \geq w_2$ and $O_1 \subsetneq X_1$.
		\end{claim}
		\begin{proof}
			By definition and assumption we have
			\begin{equation*}
				\frac{A}{w_1} \le \frac{B}{w_2} \leq \frac{c_1(O_2)}{w_2} = \frac{1-A}{w_2},
			\end{equation*}
			which implies $A \leq  w_1$. 
			The above inequality also implies that agent $1$ does not envy agent $2$ in allocation $\mathbf{O}$. 
			Since $\mathbf{O}$ is not WEF1, agent $2$ must envy agent $1$, i.e.,
			\begin{equation*}
				\frac{B}{w_2} > \frac{c_2(O_1)}{w_1} = \frac{1-B}{w_2},
			\end{equation*}
			which implies $B \geq w_2$.
			For the last property, if $O_1$ is not a proper subset of $X_1$, then we have $X_1 \subseteq O_1$ and $O_2\subseteq X_2$.
			Since agent $2$ is WEF1 towards agent $1$ in allocation $\bX$, agent $2$ should also be WEF1 towards agent $1$ in $\mathbf{O}$, which implies that $\mathbf{O}$ is WEF1 and is a contradiction.
		\end{proof}
		
		\begin{claim}
			We have $\sc(\bX) \leq 1-\frac{1-A-B}{B}\cdot \frac{w_2}{w_1} \cdot (1-B)$
		\end{claim}
		\begin{proof}
			Given Claim~\ref{claim:A-B-bound-f-at-right}, we have $O_1\subsetneq X_1$ and $X_2 \subsetneq O_2$.
			Therefore we have
			\begin{equation*}
				c_1(X_1) = 1-c_1(X_2) \le 1-\frac{c_1(O_2)}{c_2(O_2)}\cdot c_2(X_2),
			\end{equation*}
			where the inequality holds because $X_2$ is a proper subset of $O_2$ and it contains the items in $O_2$ with maximum cost-ratios.
			Furthermore, by definition of $f$ we have 
			\begin{equation*}
				c_2(X_2) = c_2(R(f+1))>\frac{w_2}{w_1}\cdot c_2(L(f-1)) \geq \frac{w_2}{w_1}\cdot c_2(O_1).
			\end{equation*}
			
			Putting the bounds together, we get
			\begin{equation*}
				\begin{aligned}
					c_1(X_1) + c_2(X_2) 
					& \leq 1-\frac{c_1(O_2)}{c_2(O_2)}\cdot c_2(X_2) + c_2(X_2) = 1- \left( \frac{c_1(O_2)}{c_2(O_2)}-1 \right) \cdot c_2(X_2)\\
					&\le 1- \left( \frac{c_1(O_2)}{c_2(O_2)}-1 \right) \cdot \frac{w_2}{w_1} \cdot c_2(O_1) = 1 - \frac{1-A-B}{B} \cdot \frac{w_2}{w_1} \cdot (1-B),
				\end{aligned}
			\end{equation*}
			where the second inequality holds because $c_1(O_2) \ge c_2(O_2)$ and $c_2(X_2) > \frac{w_2}{w_1} \cdot c_2(O_1)$.
		\end{proof}
		
		Given the above claim, letting $C = \frac{w_2}{w_1}$, we have
		\begin{equation*}
			\frac{\sc(\bX)}{\sc(\mathbf{O})} \le \frac{1-\frac{1-A-B}{B}\cdot C \cdot (1-B)}{A+B} 
			= \frac{C}{A+B} \cdot \left( \frac{1}{C}-\frac{1}{B} + 1 \right) + \frac{C(1-B)}{B}.
		\end{equation*}
		
		By Claim~\ref{claim:A-B-bound-f-at-right} we have $B \geq w_2$, which imples
		\begin{equation*}
			\frac{1}{C} - \frac{1}{B} + 1 = \frac{1}{w_2} - \frac{1}{B} \geq 0.
		\end{equation*}
		
		Hence the upper bound on $\frac{\sc(\bX)}{\sc(\mathbf{O})}$ is maximized when $A = 0$, which gives
		\begin{align*}
			\frac{\sc(\bX)}{\sc(\mathbf{O})} & \leq \frac{C}{B} \left( \frac{1}{C}-\frac{1}{B} + 1 \right) + \frac{C(1-B)}{B} = \frac{2C+1}{B} - \frac{C}{B^2} - C \\
			& = \left( \frac{2C+1}{2\sqrt{C}} \right)^2 - \left( \frac{2C+1}{2\sqrt{C}} - \frac{\sqrt{C}}{B} \right)^2 - C \\
			& \leq \left( \frac{2C+1}{2\sqrt{C}} \right)^2 - C = 1 + \frac{1}{4C} = 1 + \frac{w_2}{4\cdot w_1} \leq \frac{4+\alpha}{4},
		\end{align*}
		where the last inequality follows from the definition $\alpha = \frac{\max\{w_1,w_2\}}{\min\{w_1,w_2\}}$.
		Hence we have $\sc(\bX) \leq \frac{4+\alpha}{4}\cdot \sc(\mathbf{O})$, and the proof is complete.
	\end{proof}

	\section{Conclusion} \label{sec:conclusion}
	
	In this paper, we consider the fairness notion of weighted EF1 and try to paint a complete picture of WEF1 for the allocation of indivisible chores.
	We show that WEF1 allocations always exist for chores and propose a polynomial time algorithm to compute one, based on the weighted picking sequence algorithm.
	We further consider the picking sequences that satisfy other fairness notions, e.g. WEF$(x,y)$ with $x+y\geq 1$.
	We also consider allocations that are fair and efficient, by showing that WEF1 and PO allocations exist for bi-valued instances.
	Finally, we consider the price of fairness regarding WEF1 and provide a tight characterization of the price of WEF1 for two agents.
	
	Our work leaves many interesting problems open.
	For example, whether WEF1 allocations exist in the mixed manner (mixture of goods and chores) remains unknown.
	It is also unknown whether (weighted) EF1 and PO allocations exist for the allocation of chores, when agents have general additive cost functions.
	We believe that it would also be an interesting direction to explore the existence of WEF1 allocations when agents have cost functions beyond additive. 
	
	\bibliographystyle{abbrv}
	\bibliography{wef1}
	
	\newpage
	\appendix	
	
	\section{Analysis of Other Picking Sequence Algorithms}
	
	\subsection{Allocation of Goods}\label{ssec:wef1-goods}
	In this section, we apply our continuous perspective to the allocation of goods.
	We provide an alternative and simpler proof of Theorem 3.3 in~\cite{journals/teco/ChakrabortyISZ21}.
	In the allocation of goods, each agent has an additive valuation function $v_i: 2^M \to \bR^+ \cup \{0\}$, and agents want to maximize the utility of their bundles.
	'
	\begin{definition}[WEF1 for Goods]
		An allocation is \emph{weighted envy-free up to one good} (WEF1) if for any $i,j\in N$, there exists an item $e\in X_j$ such that
		\begin{equation*}
			\frac{v_i(X_i)}{w_i} \geq \frac{v_i(X_j-e)}{w_j}.
		\end{equation*}
	\end{definition}
	
	\begin{lemma}\label{lemma:wef1-goods}
		The weighted picking sequence protocol~\cite{journals/teco/ChakrabortyISZ21} computes WEF1 allocations for the allocation of goods in polynomial time. 
	\end{lemma}
	\begin{proof}
		Fix any two agents $i,j\in N$, we show that agent $i$ is WEF1 towards $j$.
		Let $X_i = \{e_1, \cdots, e_k\}$ and $X_j = \{e'_1, \cdots, e'_{k'}\}$ be the bundles agent $i$ and $j$ receive in the final allocation, respectively.
		Similar to our previous analysis, we assume the items are  ordered in increasing order of the index of rounds in which they are chosen.
		In the following we show that $\frac{v_i(X_i)}{w_i} \geq \frac{v_i(X_j-e'_1)}{w_j}$.
		
		We define a continuous non-increasing function $\rho: (0,k/w_i] \rightarrow \bR^+$ such that $\rho(\alpha) = v_i(e_z)$, for $\alpha\in \left(\frac{z-1}{w_i}, \frac{z}{w_i} \right]$, where $z\in \{1,2,\ldots,k\}$.
		Similarly, we define $\rho': (0, k'/w_j]  \rightarrow \bR^+$ be a continuous function: $\rho'(\alpha) = v_i(e'_z)$, for $\alpha\in \left(\frac{z-1}{w_j}, \frac{z}{w_j} \right]$, where $z\in \{1,2,\ldots,k'\}.$
		By definition of $\rho$ and $\rho'$, we have
		\begin{equation*}
			\frac{c_i(X_i)}{w_i} = \int_{0}^{\frac{k}{w_i}} \rho (\alpha) d\alpha, \quad \text{and} \quad
			\frac{c_i(X_j - e'_1 )}{w_j} = \int_{\frac{1}{w_j}}^{\frac{k'}{w_j}} \rho' (\alpha) d\alpha.
		\end{equation*}
		
		Next we establish two useful technical claims to show that $\frac{v_i(X_i)}{w_i} \geq \frac{v_i(X_j-e'_1)}{w_j}$.
		
		\begin{claim}\label{claim:k-and-k'}
			We have $k/w_i \ge (k'-1)/w_j$.
		\end{claim}
		\begin{proof}
			Consider the moment in time $t'$ when $s_j(t)$ starts to grow from $(k'-1)/w_j$, i.e. $t' = \max \{t :\; s_j(t) = (k'-1)/w_j\}$.
			At time $t'$, since $s_j(t')$ is chosen to grow, we must have that $s_j(t') \leq s_i(t')$.
			Therefore we have $k/w_i \geq s_i(t') \geq s_j(t') = (k'-1)/w_j$, as claimed.
		\end{proof}
		
		\begin{claim}\label{claim:rho-and-rho'}
			For all $\alpha \in \left(\frac{1}{w_j},\frac{k'}{w_j}\right]$, we have $
			\rho \left(\alpha - \frac{1}{w_j}\right) \ge \rho' (\alpha)$.
		\end{claim}
		\begin{proof}
			Fix any $\alpha$ and suppose that $\rho'(\alpha) = v_i(e'_z)$, i.e., $\alpha \in \left(\frac{z-1}{w_j},\frac{z}{w_j}\right]$.
			Let $t_2$ be the minimum such that $s_j(t_2) = \alpha$.
			Let $t^*$ be the maximum integer that is smaller than $t_2$.
			By definition, from time $t^*$ to $t^*+1$, $s_j(t)$ grows from $(z-1)/w_j$ to $z/w_j$ and $t_2 \in (t^*, t^*+1]$.
			Let $t_1$ be the minimum such that $s_i(t_1) = \alpha - 1/w_j$.
			By definition we have $\rho(\alpha - 1/w_j) = v_i(e_x)$, where $e_x$ is the item agent $i$ is consuming at time $t_1$.
			Since at time $t^*$, $s_j(t)$ is chosen to grow, we have $s_j(t^*) = (z-1)/w_j \leq s_i(t^*)$.
			Since $\alpha\in \left(\frac{z-1}{w_j},\frac{z}{w_j}\right]$, we have $\alpha - 1/w_j \leq (z-1)/w_j$.
			Recall that $t_1$ is the minimum such that $s_i(t_1) = \alpha-1/w_j$.
			Since $s_i(t^*) \geq (z-1)/w_j  \geq \alpha - 1/w_j$ and $s_i(t)$ is non-decreasing, we have $t_1 \leq t^*$.
			Since $t_2 \in (t^*, t^*+1]$, we have $t_1 \leq t^* < t_2$.
			In other words, the event that ``agent $i$ picks item $e_x$'' happens strictly earlier than the event ``agent $j$ picks item $e'_z$''.
			Since agent $i$ picks item $e_x$ when $e'_z$ is still available, we have $v_i(e_x) \geq v_i(e'_z)$, which implies $\rho(\alpha - 1/w_j) = v_i(e_x) \geq v_i(e'_z) = \rho'(\alpha)$ and concludes the proof.         
		\end{proof}
		
		Combining Claim~\ref{claim:k-and-k'} and~\ref{claim:rho-and-rho'}, we have
		\begin{align*}
			\frac{c_i(X_i)}{w_i} 
			& = \int_{0}^{\frac{k}{w_i}} \rho(\alpha) d \alpha
			\ge \int_{0}^{\frac{k'-1}{w_j}} \rho(\alpha) d \alpha
			\ge \int_{0}^{\frac{k'-1}{w_j}} \rho' \left( \alpha+\frac{1}{w_j} \right) d\alpha \\
			&= \int_{\frac{1}{w_j}}^{\frac{k'}{w_j}} \rho'(\alpha) d \alpha 
			= \frac{c_i(X_j - e'_1)}{w_j},
		\end{align*}
		where the first inequality follows from Claim~\ref{claim:k-and-k'} and the second inequality follows from Claim~\ref{claim:rho-and-rho'}.
		Hence agent $i$ is WEF1 towards agent $j$, which finishes the proof.
	\end{proof}

	\subsection{Computation of WEF(\texorpdfstring{$x,y$}{}) Allocations for Chores}\label{sec:wefxy}
	
	In this section, we propose an algorithm that computes WEF$(x,y)$ allocations for chores, for any $x+y \geq 1$, that is similar to Algorithm~\ref{alg:RWPS}.
	We first construct a forward sequence and then let agents follow the reversed sequence to pick their favorite remaining items.
	At the beginning of the algorithm, we set $s_i = 0$ for all $i\in N$.
	And we decide the forward sequence $(\sigma(1), \cdots, \sigma(m))$ by the following steps: for each $l\in \{1,2,\cdots,m\}$ we set $\sigma(l) \gets i^*$ where $i^*$ satisfies that $s_{i^*} + \frac{1-x}{w_i} \leq s_j + \frac{y}{w_j}$ for any $j\neq i^*$, and then update $s_{i^*} \gets s_{i^*} + \frac{1}{w_{i^*}}$.
	We argue (in the proof of Lemma~\ref{lemma:WEF(x,y)-exist-i*}) that such an agent always exists when $x+y \geq 1$.
	Then by Theorem~\ref{the:wef(x,y)}, we can show that the resulting allocation is WEF($x,y$).
	
	\begin{algorithm}[htbp]
		\caption{Reversed Weighted Picking Sequence Algorithm for WEF$(x,y)$ } \label{alg:RPSA}
		\KwIn{An instance $<M, N, \bw, \bc>$ with additive cost functions.}
		Initialize $X_i \gets \emptyset$ and $s_i \gets 0$ for all $i\in N$, and $P \gets M$\;
		\For{$t \in \{1, 2, \cdots, m\}$}{
			let $i^*$ be any agent such that $s_{i^*} + \frac{1-x}{w_{i^*}}\le s_j + \frac{y}{w_j}$ for any $j\in N\setminus \{i^*\}$ \label{line:i*-exists} \;
				$\sigma(t) \gets i^*$, $s_{i^*} \gets s_{i^*}+\frac{1}{w_{i^*}}$\;
			}
			\For{$t\in \{m, m-1, \cdots, 1\}$}{
				$i \gets \sigma(t)$, $e^* \gets \arg \min_{e \in P} \{ c_i(e) \}$, breaking ties arbitrarily\;
				$X_i \gets X_i + e^*$, $P \gets P\setminus \{e^*\}$\;
			}
			\KwOut{$\bX = \{X_1, X_2 ,\cdots, X_n\}$.}
		\end{algorithm}
		
		\begin{lemma} \label{lemma:WEF(x,y)-exist-i*}
			Algorithm~\ref{alg:RPSA} computes WEF$(x,y)$ allocations for all $x,y\in [0,1]$ with $x+y\geq 1$.
		\end{lemma}
		\begin{proof}
			We show that the agent $i^*$ specified in line~\ref{line:i*-exists} of Algorithm~\ref{alg:RPSA} always exists.
			In particular, let
			\begin{equation*}
				i^* = \arg\min_{i\in N} \left\{ s_i + \frac{1-x}{w_i} \right\}.
			\end{equation*}
			
			Since $x+y \geq 1$, for all $j\in N$ we have
			\begin{equation*}
				s_{i^*} + \frac{1-x}{w_{i^*}}\le s_j + \frac{1-x}{w_j}\le s_j + \frac{y}{w_j}.
			\end{equation*}
			
			Therefore the agent $i^*$ specified in line~\ref{line:i*-exists} always exists and the algorithm is well-defined.
			Hence for any prefix of length $t$ of the forward picking sequence, we have $s_i - \frac{x}{w_i} \leq s_j + \frac{y}{w_j}$ for any $i,j\in N$.
			Then by Theorem~\ref{the:wef(x,y)} we have that Algorithm~\ref{alg:RPSA} computes WEF$(x,y)$ allocations in polynomial time.
		\end{proof}
		
		We remark that $x+y\ge 1$ is necessary for ensuring that the above algorithm is well-defined.
		Suppose otherwise ($x+y < 1$), then the agent $i^*$ does not exist when $t = 1$, when all agents have the same weight, because for all $i,j \in N$ we always have $\frac{1-x}{w_i} > \frac{y}{w_j}$.

	\section{WEF1 and PO for Two Agents} \label{sec:PO-twoagents}
		When there are only two agents, the weighted adjusted winner algorithm~\cite{journals/teco/ChakrabortyISZ21} computes a WEF1 and PO allocation for goods in polynomial time.
		In this section, we show that there exists an algorithm that computes a WEF1 and PO for instances of chores with two agents in polynomial time. 
		\begin{theorem}\label{the:wef1+po-two}
			For the weighted allocation of chores with two agents, there exists an algorithm that can always compute a WEF1 and PO allocation for chores in polynomial time.
		\end{theorem}
		\begin{proof}
			Given any instance of chores $\mathcal{I} = <M, N, \bw, \bc>$, we construct a corresponding instance of goods $\mathcal{I}' = <M, N, \bw', \bv>$ while the valuation functions hold $v_i(e) = c_i(e)$ for any $i\in N, e\in M$, and the weights hold $w'_1 = w_2, w'_2 = w_1$.
			Note that WEF1 and PO allocations have been proven to exist for two agents~\cite{journals/teco/ChakrabortyISZ21}.
			Let $\bX' = \{X'_1, X'_2\}$ be a WEF1 and PO allocation for $\mathcal{I}'$.
			Then we compute an allocation $\bX$ for chores by $X_1 = X'_2, X_2 = X'_1$.
			We argue the allocation $\bX$ is WEF1 and PO.
			
			Note that for the cases of two agents, we have $M = X_1 \cup X_2$, leading to $X_1 = M\setminus X'_1, X_2 = M\setminus X'_2$.
			We first show that the allocation $\bX$ is PO.
			Assume otherwise and there exists another allocation $\bX^* = \{X_1^*, X_2^*\}$ that dominates $\bX$.
			We assume w.l.o.g that $c_i(X_1^*) < c_1(X_1), c_2(X_2^*) \leq c_2(X_2)$.
			Then for the instance $\mathcal{I}'$, we have $v_1(M\setminus X_1^*) > v_1(M\setminus X_1) = v_1(X'_1)$ and $ v_2(M\setminus X_2^*) \geq v_2(M\setminus X_2) = v_2(X'_2)$, which contradict the fact that $\bX'$ is PO.
			Next, we show that the allocation $\bX$ is WEF1.
			Note that in the allocation $\bX'$, for any $i\neq j$ there exists an item $e\in X_j$ such that
			\begin{equation*}
				\frac{v_i(X'_i)}{w'_i} \geq \frac{v_i(X'_j - e)}{w'_j} \Rightarrow \frac{c_i(X_j)}{w_j} \geq \frac{c_i(X_i - e)}{w_i}.
			\end{equation*}
			The deduction holds since $X_i = X'_j$ and $w'_i = w_j$.
			Hence the allocation $\bX$ is WEF1 and PO for the instance of chores $\mathcal{I}$. 
		\end{proof}
		
	\section{WEF1 and Other Fairness Notions}\label{sec:relationships}
		
		%
		In this section, we discuss the relation between the fairness notion of WEF1 and other notions including \emph{weighted proportional up to one item} (WPROP1) and \emph{AnyPrice Share} (APS).
		We show that WEF1 implies WPROP1 and $(2-\min_{i\in N} w_i)$ approximation of APS.

		
		
		\begin{definition}[WPROP1]
			An allocation is called weighted proportional up to one item (WPROP1) if for any $i\in N$, there exists an item $e\in X_i$ such that 
			\begin{equation*}
				c_i(X_i - e) \leq w_i \cdot c_i(M).
			\end{equation*}
		\end{definition}
		
		It is well known that EF1 implies PROP1 for both allocations of goods and chores, in the unweighted setting.
		However, when agents have general weights, Chakraborty et al.~\cite{journals/teco/ChakrabortyISZ21} show that WEF1 allocations are not necessarily WPROP1 for the allocation of goods.
		In contrast, we show in the following that any WEF1 allocation is WPROP1 for the allocation of chores.
		
		\begin{lemma}\label{lemma:wprop1}
			Any WEF1 allocation is WPROP1, but not vice versa.
		\end{lemma}
		\begin{proof}
			Let $\bX$ be any WEF1 allocation.
			Fix any agent $i$ and let $e^*= \arg\max_{e\in X_i} \{c_i(e)\}$ be the item with maximum cost in $X_i$.
			By the definition of WEF1, for all $j\in N$ we have
			\begin{equation*}
				\frac{c_i(X_i-e^*)}{w_i} \leq \frac{c_i(X_j)}{w_j} \quad \Rightarrow \quad
				\frac{w_j}{w_i}\cdot c_i(X_i-e^*) \leq c_i(X_j).
			\end{equation*}
			
			Therefore we have
			\begin{align*}
				c_i(X_i - e^*) & = w_i\cdot \frac{\sum_{j\in N} w_j}{w_i} \cdot c_i(X_i - e^*) \leq w_i\cdot  \sum_{j\in N} c_i(X_j) = w_i\cdot c_i(M),
			\end{align*}
			which implies that the allocation is WPROP1.
			
			Finally, via the following simple example, we show that WPROP1 can not guarantee (any bounded approximation of) WEF1, even for two symmetric agents with identical cost functions.
			
			\begin{table}[htbp]
				\centering
				\begin{tabular}{c|c|c|c}
					&  $e_1$ & $e_2$ & $e_3$ \\ \hline
					Agent 1   & $\epsilon$ & $\boxed{1}$ & $\boxed{1}$ \\
					Agent 2   & $\boxed{\epsilon}$ & $1$ & $1$
				\end{tabular}
				\caption{Example showing that WPROP1 does not imply WEF1, where $\epsilon > 0$ is arbitrarily small.}
				\label{tab:hard-for-prop1}
			\end{table}
			
			It can be easily observed that the allocation highlighted by the boxed items is WPROP1\footnote{In fact the allocation is WPROPX, a fairness notion that requires $c_i(X_i - e) \leq w_i\cdot c_i(M)$ for all $i\in N$ and $e\in X_i$.}, but is not WEF1 because $c_1(X_1 - e) \geq 1/\epsilon\cdot c_1(X_2)$, for all $e\in X_1$.
		\end{proof}
		
		
		\medskip
		
		Next we study the fairness notion of AnyPrice Share (APS) fair.
		The notion is introduced by Babaioff et al.~\cite{conf/sigecom/BabaioffEF21} for the allocation of goods, and is then adapted to the case of chores in~\cite{conf/www/0037L022,journals/corr/abs-2211-13951}.
		
		Fix any agent $i\in M$.
		We call $r = (r_1, r_2, \cdots, r_m)$ a \emph{reward vector} if $r_e \geq 0$ for all $e\in M$ and  $r(M) = \sum_{e\in M} r_e = c_i(M)$. 
		Let $R$ be the set of all reward vectors.
		The AnyPrice Share $\APS_i$ of agent $i$ is defined as the minimum value such that no matter how the reward vector is set, agent $i$ can always find a subset of chores with total reward at least $w_i$ and total cost at most $\APS_i$.
		
		\begin{definition}[APS]
			The AnyPrice Share (APS) of agent $i$ with weight $w_i$ is defined as 
			\begin{equation*}
				\APS_i = \max_{r\in R} \min_{S\subseteq M} \left\{c_i(S) : \sum_{e\in S} r_e \geq w_i\cdot c_i(M) \right\}.
			\end{equation*}
		\end{definition}
		
		An allocation is called $\alpha$-APS allocation if for all $i\in N$, it holds that $c_i(X_i) \leq \alpha\cdot \APS_i$.	
		
		\begin{lemma}\label{lemma:2-aps}
			If an allocation $\bX$ is WEF1, then it is $(2 - \min_{i\in N} w_i)$-APS.
			However, APS allocations cannot guarantee constant approximation of WEF1.
		\end{lemma}
		\begin{proof}
			Fix any agent $i$. 
			We show that $\bX$ is $(2 - w_i)$-APS to $i$.
			As it is commonly observed, e.g., see~\cite{conf/www/0037L022}, we have the following property regarding $\APS_i$:
			\begin{equation*}
				\APS_i \geq \max \left\{ w_i\cdot c_i(M), \ \max_{e\in M} \{c_i(e)\} \right\}.
			\end{equation*}
			
			Let $e^*= \arg\max_{e\in X_i} \{c_i(e)\}$ be the item with maximum cost in $X_i$.
			Since $\bX$ is WEF1, we have 
			\begin{align*}
				& c_i(X_i - e^*) = w_i\cdot \frac{\sum_{j\in N} w_j}{w_i} \cdot c_i(X_i - e^*) \\
				\leq \ & w_i\cdot  \sum_{j\in N\setminus \{i\}} c_i(X_j) + w_i\cdot c_i(X_i - e^*) = w_i\cdot c_i(M) - w_i\cdot c_i(e^*).
			\end{align*}
			
			Hence we have
			\begin{align*}
				c_i(X_i) \leq  w_i \cdot c_i(M) + (1- w_i) \cdot c_i(e^*)
				\leq \APS_i + (1-w_i)\cdot \APS_i = (2-w_i)\cdot \APS_i,
			\end{align*}
			where the second inequality follows because $\APS_i \geq w_i\cdot c_i(M)$ and $\APS_i \geq c_i(e^*)$.
			
			\medskip
			
			Next, we give an example instance showing that the approximation ratio $(2-\min_{i\in N} w_i)$ is tight.
			Consider the following instance with $n$ symmetric agents with identical cost functions.
			
			\begin{table}[htbp]
				\centering
				\begin{tabular}{c|c|c|c|c|c|c|c|c}
					Agents &  $e_1$ & $e_2$ & $e_3$ & $\cdots$ & $e_n$ & $e_{n+1}$ & $\cdots$ & $e_{2n-1}$ \\ 
					\hline
					${1}$   & $\boxed{\frac{1}{n}}$ & $\boxed{\left( 1-\frac{1}{n} \right)\cdot \frac{1}{n}}$ & $\left( 1-\frac{1}{n} \right)\cdot \frac{1}{n}$ & $\cdots$ & $\left( 1-\frac{1}{n} \right)\cdot \frac{1}{n}$ & $\frac{1}{n^2}$ & $\cdots$ & $\frac{1}{n^2}$ \\
					${2}$   & $\frac{1}{n}$ & $\left( 1-\frac{1}{n} \right)\cdot \frac{1}{n}$ & $\boxed{\left( 1-\frac{1}{n} \right)\cdot \frac{1}{n}}$ & $\cdots$ & $\left( 1-\frac{1}{n} \right)\cdot \frac{1}{n}$ & $\frac{1}{n^2}$ & $\cdots$ & $\frac{1}{n^2}$ \\
					$\cdots$ & $\cdots$ & $\cdots$ & $\cdots$ & $\cdots$ & $\cdots$ & $\cdots$ & $\cdots$ & $\cdots$ \\
					${n-1}$ & $\frac{1}{n}$ & $\left( 1-\frac{1}{n} \right)\cdot \frac{1}{n}$ & $\left( 1-\frac{1}{n} \right)\cdot \frac{1}{n}$ & $\cdots $ & $\boxed{\left( 1-\frac{1}{n} \right)\cdot \frac{1}{n}}$ & $\frac{1}{n^2}$ & $\cdots$ & $\frac{1}{n^2}$ \\
					${n}$   & $\frac{1}{n}$ & $\left( 1-\frac{1}{n} \right)\cdot \frac{1}{n}$ & $\left( 1-\frac{1}{n} \right)\cdot \frac{1}{n}$ & $\cdots $ & $\left( 1-\frac{1}{n} \right)\cdot \frac{1}{n}$ & $\boxed{\frac{1}{n^2}}$ & $\cdots$ & $\boxed{\frac{1}{n^2}}$ \\
				\end{tabular}
				\caption{Instance showing that the approximation ratio $2 - \min_{i\in N} w_i = 2-\frac{1}{n}$ is tight.}
				\label{tab:hard-for-aps}
			\end{table}
			
			As shown in Table~\ref{tab:hard-for-aps}, there are $2n-1$ items to be allocated to $n$ symmetric agents, for which we have $\APS_i = \APS \geq 1/n$ for all $i\in N$.
			In fact, we have $\APS = 1/n$ because given any reward vector, one of the bundles $\{ e_1 \}, \{ e_2, e_{n+1} \}, \{ e_3, e_{n+2} \}, \ldots, \{ e_n, e_{2n-1} \}$ must have reward at least $1/n$, and can be chosen to ensure a cost of $1/n$ when defining $\APS$.
			Consider the allocation $\bX$ indicated by the boxed items.
			Since all agents other than $1$ have cost $\left(1-\frac{1}{n}\right)\cdot \frac{1}{n}$ and agent $1$ has cost $\left(1-\frac{1}{n}\right)\cdot \frac{1}{n}$ after removing $e_1$, the allocation is EF1.
			However, the allocation is not better than $\left(2-\frac{1}{n}\right)$-APS because $c_1(X_1) = \left(2-\frac{1}{n}\right)\cdot \APS$.

			\medskip
			
			Finally, we provide a simple example showing that APS allocations do not guarantee any constant approximation of WEF1, even for the unweighted setting.
			Consider an instance with four items $\{e_1, e_2, e_3, e_4\}$ and three identical agents with cost function $c = (1, 1/2, 1/2, \epsilon)$, where $\epsilon > 0$ is arbitrarily small.
			By allocating $e_1$ to agent $1$, $\{e_2,e_3\}$ to agent $2$ and $e_4$ to agent $3$, we have an APS allocation $\bX$ since $\APS \geq c(e_1) = 1$ for all $i\in N$.
			However, this allocation is far from being EF1 since $c(X_2-e) = 1/2 \geq \omega(1)\cdot \epsilon = \omega(1)\cdot c(X_3)$ for all $e\in X_2$.
		\end{proof}
	\end{document}